\documentclass[peerreview,12pt]{IEEEtran}
\usepackage{graphicx}
\usepackage{times,amsmath,amssymb,array,stfloats,midfloat,psfrag,setspace,multirow}
\usepackage[noadjust]{cite}
\usepackage{mathrsfs}

\newtheorem{theorem}{\textbf{Theorem}}
\newtheorem{definition}{\textbf{Definition}}
\newtheorem{lemma}{\textbf{Lemma}}
\newtheorem{corollary}{\textbf{Corollary}}
\newtheorem{remark}{Remark}
\newcommand{\R}{\mathcal{R}}
\newcommand{\bR}{\mathbf{R}}
\newcommand{\uR}{\underline{\mathcal{R}}}
\newcommand{\oR}{\overline{\mathcal{R}}}
\newcommand{\MLD}{_{\textrm{MLD}}}
\newcommand{\MD}{_\textrm{MD}}

\newcommand{\e}{\varepsilon}
\newcommand{\V}{\tilde{V}}

\renewcommand{\L}{\mathscr{L}}
\newcommand{\N}{\mathcal{N}}
\newcommand{\E}{\mbox{${\mathbb E}$}}

\newcommand{\D}{\mathbf{D}}

\renewcommand{\S}{\mathcal{S}}

\title{Asymmetric Multilevel Diversity Coding and Asymmetric Gaussian Multiple Descriptions}
\date{}
\begin{document}

\author{Soheil Mohajer,~\IEEEmembership{Student Member,~IEEE}, Chao Tian,~\IEEEmembership{Member,~IEEE},  and\\ Suhas N. Diggavi,~\IEEEmembership{Member,~IEEE}
\thanks{S. Mohajer and S. N. Diggavi are with the School of Computer and Communication Sciences, Ecole Polytechnique Federale de Lausanne, Switzerland. C. Tian  is with  
AT\&T Labs-Research, Florham Park, New Jersey, USA.}
}

\maketitle

\begin{abstract}
We consider the asymmetric multilevel diversity (A-MLD) coding problem, 
where a set of $2^K-1$ information sources, ordered in a decreasing 
level of importance, is encoded into $K$ messages (or descriptions). 
There are $2^K-1$ decoders, each of
which has access to a non-empty subset of the encoded messages. Each
decoder is required to reproduce the information sources up to a 
certain importance level depending on the combination of descriptions available to it. We obtain a single letter 
characterization of the achievable rate region for the
$3$-description problem. In contrast to symmetric multilevel diversity coding,
source-separation coding is not sufficient in the asymmetric
case, and ideas akin to network coding need to be used strategically.
Based on the intuitions gained in treating the A-MLD problem, we derive
inner and outer bounds for the rate region of the asymmetric Gaussian multiple description (MD) 
problem with three descriptions. Both the inner and outer bounds have a similar geometric structure to 
the rate region template of the A-MLD coding problem, and moreover, we show that 
the gap between them is small, which results in an approximate 
characterization of the asymmetric Gaussian three description rate region. 
\end{abstract}

\section{Introduction}

In the symmetric multilevel diversity coding (MLD) problem
\cite{RocheYeungHau:97}, $K$ source sequences are encoded into $K$
descriptions, which are sent to the decoders through noiseless
channels. These source sequences have a decreasing levels of importance, and 
each decoder has access to a non-empty
subset of the descriptions. The goal of the encoder is to produce the
descriptions such that each decoder with $k$ available descriptions is 
able to reconstruct the $k$ most important source sequences.
The symmetric
MLD problem was motivated by fault-tolerant
storage for disk arrays and for incremental priority encoding on
packet erasure channels; see \cite{RocheYeungHau:97} for more details. The MLD 
problem with three levels was solved by Roche \emph{et al.} in \cite{RocheYeungHau:97}, and 
the result was later extended by Yeung and Zhang \cite{YeungZhang:99} to 
an arbitrary number of levels. It was shown that source-separation coding\footnote{This was called 
 superposition coding in these papers. In order not to confuse this
  with the common terminology of broadcast channels, the new
  terminology has been adopted here, as suggested by R. Yeung.} 
is optimal for the symmetric problem. This means that each source
sequence can be compressed separately, and then the descriptions are obtained
 by concatenating the compressed source sequences appropriately.

In this work we formulate the asymmetric multilevel diversity (A-MLD) coding
problem. The problem can be understood as a refined version of symmetric MLD coding problem, and 
it is naturally applicable in distributed disk storage applications 
with asymmetric (unequal) reliabilities, in contrast to
symmetric (equal) reliabilities which motivate the symmetric MLD problem. 
Similarly, for packet erasure applications, the erasure probabilities for the
sub-packets may not be equal because the paths over which they are sent
may have different reliabilities. As such, in both applications, we
may wish to utilize not just the number of the encoders which are
accessible, but also their identities, since the descriptions are no 
longer symmetric. Therefore, the difference between the MLD and 
A-MLD problem is that in the asymmetric version the levels of reconstruction is determined 
by the specific combination of descriptions available to them, not 
just the number of descriptions. 

More precisely, $2^K-1$ source sequences are encoded into $K$ descriptions 
at the encoder. 
The $2^K-1$ decoders are ordered in a specific way, 
and the goal of the encoder is to produce the descriptions such that the $k$-th decoder 
 is able to reconstruct the $k$ most important source sequences, for
$k=1,\dots,2^K-1$. In this work, we only consider the $3$-description
case and provide a complete characterization of the achievable rate region. In
particular we show that source-separation coding coding is \emph{not}
optimal for this problem, and the source sequences in different levels
have to be jointly encoded (like in {\em network coding}) in an
optimal coding strategy. We also show that the scheme using \emph{linear} combinations
of these compressed sequences is optimal. We note that various special 
cases of $3$-description problem were studied 
in\footnote{We would like to thank R. Yeung for bringing this work to our attention.} \cite{Hau-thesis:95}, where, however, only no more than \emph{three} information sources
were considered. The characterization we provide in this work strictly subsumes those 
considered in \cite{Hau-thesis:95}. 

Let us now turn to a closely related problem, namely the multiple description (MD) problem.
In this problem a source is mapped into $K$ 
descriptions and sent to $2^K-1$ decoders, just as in the A-MLD coding problem. 
The decoders are required to reconstruct the source sequence within certain distortions 
using the available descriptions. The MD rate region characterization is long-standing 
open problem in information theory with a long history \cite{ElgCov82,ozarow,ZhaBer87}. Despite 
many important results, the problem is still open, even for the quadratic Gaussian 
case with only three descriptions. Using the intuitions gained in treating the A-MLD problem 
as well as the sum-rate lower bound for symmetric Gaussian MD problem
recently discovered in \cite{TianMohDig08}, we develop inner and outer bounds for the MD 
rate region, both of which bear similar geometric structure to the A-MLD coding rate region. 
Moreover, the gap between the bounds is small (less than 1.3 bits in terms of the Euclidean distance between the bounding planes), yielding an approximate characterization. One surprising consequence of this result is that the proposed simple architecture based on successive 
refinement (SR) \cite{EquitzCover:91} and A-MLD
coding is in fact close to optimality. From an engineering viewpoint,
this suggests that one can design simple and flexible MD codes that
are (approximately) optimal. 

One important observation leading to this work is the intimate
connection between the multilevel diversity (MLD) coding problem and
the MD problem observed in \cite{TianMohDig08-2}. There we showed that
for the symmetric MD problem, achievable rate region based on SR
coding coupled with symmetric multilevel diversity (S-MLD) coding
provides good approximation to the MD rate region under symmetric
distortion constraints; perhaps more interestingly, the achievable
rate region has the same geometric structure as that of the symmetric
MLD coding rate region. In fact, the symmetric MLD coding result is
essential for establishing the symmetric MD result in
\cite{TianMohDig08-2}. The result in \cite{TianMohDig08-2} suggests a
general approach in treating lossy source coding problems: first solve
a corresponding a lossless version of the problem, then extend the
results and intuitions to its lossy counterpart to yield an
approximate characterization. This is exactly our motivation to
formulate the A-MLD coding problem, and indeed the result given in
this work further illustrates the effectiveness of this approach.

The paper is organized as follows. In Section II, we
introduce the notations and provide a formal definition of the problems.
In Section III, we present the main results of the paper.  
We prove the main theorem for rate region characterization of the 
A-MLD problem in Section IV. In Section V, we focus on deriving the outer
and inner bounds for the rate region of the A-MD problem. Finally, Section
VI concludes the paper. Some of the detailed and technical proofs are
given in the appendix.

\section{Notations and Problem Formulation}

In this section we provide  formal definitions for both the asymmetric multilevel diversity (A-MLD) and the asymmetric multiple description (A-MD) coding problems. Since we need to use the result of the A-MLD problem when treating the A-MD problem, we may use different notations for these problems in order to avoid confusion.  

\subsection{Asymmetric Multilevel Diversity Coding}

Let $\{(V_{1,t},V_{2,t},\dots,V_{2^K-1,t})\}_{t=1,2,\dots}$ be an independent and identically
distributed process sampled from a finite size alphabet
$\mathcal{V}_1\times\mathcal{V}_2\times\cdots\times\mathcal{V}_{2^K-1}$ with
time index $t$. This can be considered as $2^K-1$ pieces of independent data streams,
namely, $\{V_{1,t}\},\dots,\{V_{2^K-1,t}\}$, where each data stream is an independently 
and identically distributed sequence. The data streams are ordered with decreasing
importance, e.g., consecutive refinements of a single source. We use $V_i^n$ to denote a 
length $n$ sequence of $V_i$, namely, $V_i^n=(V_{i,1},\dots,V_{i,n})$.

Define the  vector random variables $U_j$ as $U_{j}\triangleq(V_{1},\dots,V_{j})$ for
$j=1,\dots,2^K-1$, and $U_{0}\triangleq 0$.  We use $U_j^n$ to denote length $n$ sequences of 
$U_j$. We may simply use  $U^n$ to denote 
$U_{2^K-1}^n=(V_{1}^n,\dots,V_{2^K-1}^n)$ for brevity. Note that $U_j^n$ is a two-dimensional array, whose elements are 
independent of each other along both directions, $i=1,\dots,j$, and $t=1,\dots,n$. 

The Shannon entropy rate of the source $V_k$ is denoted by $h_k$. We also denote the entropy of 
$U_j$ by $H_j$, where the independence of sources $V_k$'s implies
\begin{eqnarray}
H_j=H(U_j)=H(V_1,\dots,V_j)=\sum_{i=1}^j H(V_i)=\sum_{i=1}^j h_i.
\end{eqnarray}


The A-MLD problem can be described as follows. Consider $2^K-1$ source sequences which are 
fed to a single encoder. The encoder produces $K$ descriptions, denoted as 
$\Gamma_1,\Gamma_2\dots,\Gamma_K$ to encode the source sequences. The descriptions are 
sent over $K$ perfect channel. There are $2^K-1$ decoders, each has access to a non-empty subset of
the descriptions, $\S \subseteq \{\Gamma_1,\Gamma_2\dots,\Gamma_K\}$,  and wishes to decode losslessly the
source data streams below a certain \emph{level}, which is a function
of the description set $\S$. Fig.~\ref{fig:3AMLD-setup} illustrates the problem setting for $K=3$, and 
a specific decoding requirement for the decoders. 

Formally, we define the notion of
\emph{ordering level} to connect the decoding requirement of the decoders to their available description subsets as follows.
\begin{definition}
A valid \emph{ordering level} (or simply \emph{ordering}) on the non-empty subsets\footnote{For the rest of this paper, by \emph{subset} we always mean a non-empty subset although it is not precisely mentioned. } of $\{\Gamma_1,\Gamma_2\dots,\Gamma_K\}$ is a one-to-one
mapping $\L:\mathcal{P}(\{\Gamma_1,\Gamma_2\dots,\Gamma_K\})\setminus \emptyset \longrightarrow
\{1,\dots,2^K-1\}$ satisfying
\begin{description}
\item[(i)]   $\L(\{\Gamma_1\}) < \L(\{\Gamma_2\}) < \cdots <\L(\{\Gamma_K\})$,
\item[(ii)]  $\mathcal{S} \subset \mathcal{T}$ implies $\L(\mathcal{S}) < \L(\mathcal{T})$,
\end{description}
where $\mathcal{P}(M)$ is the power set of  $M$.
\label{def:order}
\end{definition}

The ordering level will be used to determine the decoding requirements 
of the decoders, \emph{e.g.,} a decoder with a set of descriptions $\S$ 
needs to decode the first $\L(\S)$ source streams. Condition (i) is given 
to avoid permuted repetition of the levels, where without loss of generality, 
we assume an initial ordering on the single description decoders. Condition 
(ii) is a natural fact that if  $\mathcal{S}$ is a subset of  $\mathcal{T}$, then 
the corresponding decoder can not do better than what decoder $\mathcal{T}$ 
can. We may simplify the notation occasionally, by omitting the braces, 
\emph{e.g.,} $\L(\Gamma_1,\Gamma_2)=\L(\{\Gamma_1,\Gamma_2\})$. 
The inverse mapping $\L^{-1}(k)$ is well defined, which is the subset of 
descriptions  whose ordering level is $k$.

An $(n; \L; M_i, i\in\{1,2,\dots,K\})$ MLD-code is defined by a set of encoding functions
\begin{eqnarray}
F_i: \mathcal{V}_1^n \times \mathcal{V}_2^n \times \cdots \times \mathcal{V}_{2^K-1}^n 
\longrightarrow \{1,2,\dots,M_i\},\quad
i\in\{1,2,\dots,K\},
\end{eqnarray}
and decoding functions
\begin{eqnarray}
\displaystyle
G_{\S}: \prod_{j: \Gamma_j\in\S}\{1,\dots,M_j\} \longrightarrow
\mathcal{V}_1^n \times \mathcal{V}_2^n \times \cdots \times \mathcal{V}_{\L(\S)}^n, \quad
 \S\subseteq\{\Gamma_1,\dots,\Gamma_K\},
\end{eqnarray}
where $\prod$ denotes a set product. We define
\begin{eqnarray}
\hat{U}_{\L(\S)}^n(\S)\triangleq G_{\S}(F_j(U^n); j:\Gamma_j\in\S )
\end{eqnarray}
and $\hat{V}_i^n(\S)$ is the corresponding part of $\hat{U}_{\L(\S)}^n(\S)$, for $i\leq \L(\S)$.

A rate tuple $\bR^{\L}=(R_1,R_2,\dots,R_K)$ is called admissible for a prescribed ordering $\L$, if for any $\e> 0$
and sufficiently large $n$, there exist an $(n; \L; M_i, i\in\{1,2,\dots,K\})$ MLD-code such that
\begin{eqnarray}
\frac{1}{n}\log M_i \leq R_i + \e,\quad i\in\{1,2,\dots,K\},
\end{eqnarray}
and
\begin{eqnarray}
\Pr(\hat{V}_i^n(\S) \neq V_i^n ) <\e \quad \forall \S \subseteq\{\Gamma_1,\dots,\Gamma_K\}, \textrm{ and }
\forall i\leq \L(\S)\label{level}.
\end{eqnarray}
The main goal in the (lossless) multilevel diversity coding problem 
is to characterize $\R\MLD$, the set of all
achievable rate tuples $(R_i; i\in\{1,2,\dots,K\})$ in terms of the entropy of the
source sequences and the given ordering level. We denote such rate region by $\R\MLD^{\L}$ for a specific ordering.

In this paper we consider this problem for three descriptions ($K=3$) and give a 
complete characterization of the rate region. It is straightforward to
show that there are eight possible orderings for $\{\Gamma_1,\Gamma_2,\Gamma_3\}$, which are shown in  Table~\ref{tbl:ordering}. We may further divide each ordering into sub-regimes to simplify the problem for each case. The results of this work are general and hold for all possible orderings. However, in order to illustrate the result,  we may specialize some of the arguments/theorems to the ordering level $\L_1$ defined as 
\begin{align*}
\begin{array}{c}
\L_1(\Gamma_1)=1,  \quad \L_1(\Gamma_2)=2,  \quad \L_1(\Gamma_3)=3,\\
\L_1(\Gamma_1,\Gamma_2)=4, \quad \L_1(\Gamma_1,\Gamma_3)=5, \quad \L_1(\Gamma_2,\Gamma_3)=6,\quad \L_1(\Gamma_1,\Gamma_2,\Gamma_3)=7.
\end{array}
\end{align*}
 The setting of the problem for the ordering level $\L_1$ is shown in Fig.~\ref{fig:3AMLD-setup}. 
\begin{figure}[t]
\begin{center}
 \psfrag{d1}[Bc][Bc]{$\ \mathtt{DEC}$}
 \psfrag{d2}[Bc][Bc]{$\ \mathtt{DEC}$}
 \psfrag{d3}[Bc][Bc]{$\ \mathtt{DEC}$}
 \psfrag{d4}[Bc][Bc]{$\ \mathtt{DEC}$}
 \psfrag{d5}[Bc][Bc]{$\ \mathtt{DEC}$}
 \psfrag{d6}[Bc][Bc]{$\ \mathtt{DEC}$}
 \psfrag{d7}[Bc][Bc]{$\ \mathtt{DEC}$}
  \psfrag{1}[Bc][Bc]{$\hspace{-13mm}\hat{V}_1^n$}
  \psfrag{2}[Bc][Bc]{$\hspace{-6mm}\hat{V}_1^n, \hat{V}_2^n$}
  \psfrag{3}[Bc][Bc]{$\hat{V}_1^n,\hat{V}_2^n,\hat{V}_3^n$}
  \psfrag{4}[Bc][Bc]{$\ \ $}
  \psfrag{5}[Bc][Bc]{$\ \ \ \vdots $}
  \psfrag{6}[Bc][Bc]{$\ $}
  \psfrag{7}[Bc][Bc]{$\hat{V}_1^n,\dots,\hat{V}_7^n$}
  \psfrag{e}[Bc][Bc]{$\mathtt{ENC}$}
 \psfrag{x1}[Bc][Bc]{$\ V_1^n$}
 \psfrag{x2}[Bc][Bc]{$\ V_2^n$}
 \psfrag{x3}[Bc][Bc]{$\vdots$}
 \psfrag{x7}[Bc][Bc]{$\ V_7^n$}
 \psfrag{a}[Bc][Bc]{$\Gamma_1$}
 \psfrag{b}[Bc][Bc]{$\Gamma_2$}
 \psfrag{c}[Bc][Bc]{$\Gamma_3$}
\hspace{2mm}
\includegraphics[width=10cm]{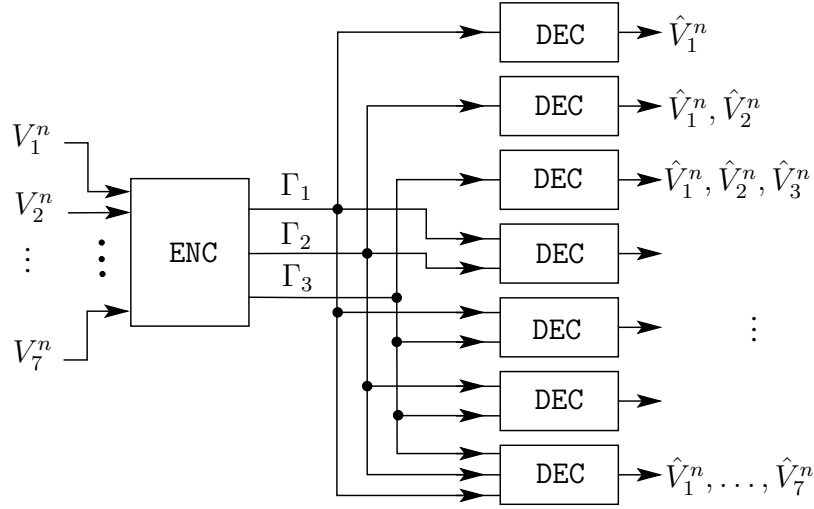}
\caption{The $3$-description asymmetric multilevel diversity coding problem for ordering level $\L_1$.}
\label{fig:3AMLD-setup}
\end{center}
\end{figure}

Fig.~\ref{desc} shows the subset of source streams
which should be recovered by each subset of descriptions in $\L_1$ setting.

\begin{figure}[t]
\begin{center}
   \psfrag{A}[Bc][Bc]{$\overbrace{\hspace{23mm}}^{\Gamma_1}$}
   \psfrag{B}[Bc][bc]{$\underbrace{\hspace{39mm}}_{\Gamma_2}$}
   \psfrag{AB}[Bc][Bc]{$\overbrace{\hspace{70mm}}^{\Gamma_1,\Gamma_2}$}
   \psfrag{ABC}[Bc][Bc]{$\underbrace{\hspace{120mm}}_{\Gamma_1,\Gamma_2,\Gamma_3}$}
   \psfrag{dot1}[Bc][Bc]{$\hspace{70mm}\cdots$}
   \psfrag{dot2}[Bc][Bc]{$\cdots$}
   \psfrag{1}[Bc][bc]{$V_1^n$}
   \psfrag{2}[Bc][bc]{$V_2^n$}
   \psfrag{3}[Bc][bc]{$V_3^n$}
   \psfrag{4}[Bc][bc]{$V_4^n$}
   \psfrag{5}[Bc][bc]{$V_5^n$}
   \psfrag{6}[Bc][bc]{$V_6^n$}
   \psfrag{7}[Bc][bc]{$V_7^n$}
   \includegraphics[width=12cm]{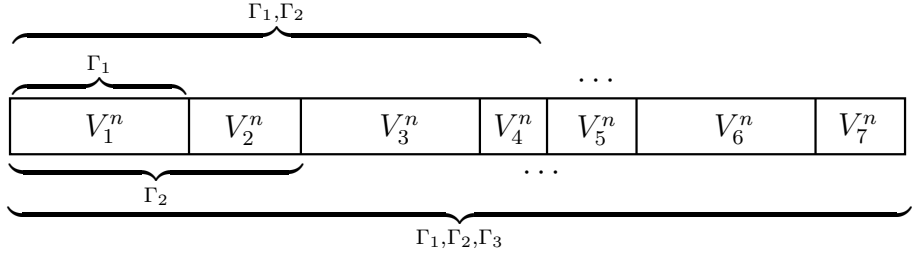}
\caption{Levels assigned to the description subsets determines the
recoverable source subsequences. The requirements corresponding to the ordering level $\L_1$ are shown in this figure.} \hspace{.1cm}
\label{desc}
\end{center}
\vspace{-5pt}
\end{figure}

\vspace{-2pt}

\subsection{Asymmetric Gaussian Multiple-Description Coding}

Let $\{X(t)\}_{t=1,2,\dots}$ be a sequence of independent and identically distributed zero mean and unit variance real-valued Gaussian source, \emph{i.e.,}  $\mathcal{X}=\mathbb{R}$, with time index $t$. Moreover, the reconstruction alphabet is also assumed to be $\mathbb{R}$. The vector $X(1),X(2),\dots,X(n)$ is denoted by $X^n$. We use capital letters for random variables, and the corresponding lower-case letters for their realization. The quality of the reconstruction is measured by the quadratic distance between the original sequence $x^n$ and the reconstructed one $\hat{x}^n$. Formally,  we define the distortion as 
\begin{eqnarray}
d(x^n, \hat{x}^n)= \frac{1}{n}\sum_{k=1}^n |x(k)-\hat{x}(k)|^2,  
\end{eqnarray}

In a general multiple description setting, the encoders produces $K$ descriptions, namely $\Gamma_1,\Gamma_2,\dots,\Gamma_K$ based on the source sequence and sends them to the decoders through noiseless channels. Each decoder receives a non-empty subset of the descriptions, and has to reconstruct the source sequence $\hat{x}^n$ which satisfies a certain level of fidelity. 

In a manner  similar to the last subsection, we denote each decoder by the corresponding set of available descriptions.
Each decoder $\S$ has a distortion constraint $D_{\S}$, and needs to reconstruct the source such that the corresponding expected distortion does not exceed this constraint. The main goal in this problem is to characterize the set of admissible rates of the descriptions in a way that such reconstructions are possible. We present a formal definition of the problem next.

An  $(n ; M_i, i\in\{1,\dots,K\}; \Delta_{\S}, \S\subseteq\{\Gamma_1,\dots,\Gamma_K\} )$ MD-code is defined as a set of encoding functions
\begin{eqnarray}
F_i: \mathcal{X}^n\longrightarrow \{1,2,\dots,M_i\},\quad
i\in\{1,2,\dots,K\},
\end{eqnarray}
and $2^K-1$ decoding functions
\begin{eqnarray}
G_{\S}: \prod_{j:\Gamma_j\in\S}\{1,\dots,M_j\} \longrightarrow
\mathcal{X}^n, \quad \S\subseteq\{\Gamma_1,\dots,\Gamma_K\},
\end{eqnarray}
with
\begin{eqnarray}
\Delta_{\S}=\E d(X^n,\hat{X}_{\S}^n),\quad  \S\subseteq\{\Gamma_1,\dots,\Gamma_K\},
\end{eqnarray}
where 
\begin{eqnarray}
\hat{X}^n_{\S}=G_{\S}(F_j(X^n), j:\Gamma_j\in\S).
\end{eqnarray}
Again, $\prod$ denotes set product, and $\E$ is the expectation operator. 

A rate tuple $\bR=(R_1,R_2,\dots,R_K)$ is called $\D=(D_{\S}; \S\subseteq\{\Gamma_1,\dots,\Gamma_K\})$-admissible if for every $\e > 0$ and sufficiently large $n$, there exists an $(n; M_i, i\in\{1,\dots,K\}; \Delta_{\S}, \S\subseteq\{\Gamma_1,\dots,\Gamma_K\})$ MD-code such that 
\begin{eqnarray}
\frac{1}{n} \log M_i \leq R_i+\e, \quad i\in\{1,\dots,K\},
\end{eqnarray}
and
\begin{eqnarray}
\Delta_{\S} \leq D_{\S}+\e, \quad \S\subseteq\{\Gamma_1,\dots,\Gamma_K\}.
\end{eqnarray}
We denote by $\R\MD(\D)$ the set of all $\D$-admissible rate tuples, which we seek to characterize. 


Let $\mathcal{T}$ and $\S$  be two description sets, satisfying $\mathcal{T} \subseteq \mathcal{S} \subseteq \{\Gamma_1,\dots,\Gamma_K\}$. It is clear that the decoder with access to $\mathcal{S}$ can reconstruct the source sequence as \emph{well} as the one with access to $\mathcal{T}$ does, even if $D_{\mathcal{S}}\geq D_{\mathcal{T}}$. The following lemma shows that slightly modification of the distortion vector in order to satisfy such property does not change the admissible rate region.

\begin{lemma}
For a given distortion vector $\D$, define $\tilde{\D}$ as $\tilde{\D}=(\tilde{D}_{\S}; \S\subseteq\{\Gamma_1,\dots,\Gamma_K\})$, where
\begin{eqnarray*}
\tilde{D}_{\S}=\min_{\mathcal{T}:\mathcal{T}\subseteq\S} D_{\mathcal{T}}.
\end{eqnarray*}
Then $\R\MD(\tilde{\D})=\R\MD(\D)$.
\label{lm:modif}
\end{lemma}
\begin{proof}[Proof of Lemma \ref{lm:modif}]
It is clear that $\tilde{D}_{\S}\leq D_{\S}$ for all $\S\subseteq\{\Gamma_1,\dots,\Gamma_K\}$, and therefore $\R\MD(\tilde{\D})\subseteq\R\MD(\D)$. So, it remains to prove $\R\MD(\D)\subseteq\R\MD(\tilde{\D})$. Let $\bR\in\R\MD(\D)$ be an admissible rate tuple for $\D$, and $(n;M_i;\Delta_{\S})$ be a code for a given $\varepsilon$ which achieves the distortion constraints $\D$, with encoding functions $\{F_i\}$ and decoding functions $\{G_{\S}\}$. We can easily modify the decoding functions and obtain a code which satisfies $\tilde{\D}$. By the definition of $\tilde{\D}$, for all $\S$ we have $\tilde{D}_{\S}=D_{\tilde{\S}}$, where 
\begin{eqnarray*}
\tilde{\S}\triangleq \arg \min_{\mathcal{T}:\mathcal{T}\subseteq\S} D_{\mathcal{T}}.
\end{eqnarray*}
Define
\begin{align*}
\tilde{X}_{\S}^n=\tilde{G}_{\S}(F_j(X^n); j:\Gamma_j\in\S)\triangleq \hat{X}_{\tilde{\S}}^n.
\end{align*}
Obviously, 
\begin{eqnarray*}
\E d(X^n,\tilde{X}_{\S}^n)=\E  d(X^n,\hat{X}_{\tilde{\S}}^n)  \leq D_{\tilde{\S}}+\e =\tilde{D}_{\S} +\e.
\end{eqnarray*}
Thus the similar code with the modified decoding functions satisfies the constraint tuple $\tilde{\D}$, and therefore $\bR\in\R\MD(\tilde{\D})$.
\end{proof}

Given this lemma, we can assume, without loss of generality, that $D_{\mathcal{T}}\leq D_{\S}$ for all $\S \subseteq\mathcal{T}$. These distortion constraints then induce an ordering on the decoders, or equivalently on their associated subset of descriptions.

 
In this work, again we focus on the three description ($K=3$) problem, and present the results in general form, \emph{i.e.,}  regardless the exact ordering. Occasionally we shall provide the proof details only for the specific sorted distortion constraints 
\begin{eqnarray*}
D_{\Gamma_1} \geq D_{\Gamma_2} \geq  D_{\Gamma_3} \geq D_{\Gamma_1 \Gamma_2} \geq D_{\Gamma_1 \Gamma_3} \geq D_{\Gamma_2 \Gamma_3} \geq D_{\Gamma_1 \Gamma_2 \Gamma_3},
\end{eqnarray*}
which induces the ordering 
\begin{eqnarray*}
\L(\Gamma_1) < \L(\Gamma_2) < \L(\Gamma_3) <\L(\Gamma_1 \Gamma_2) <\L(\Gamma_1 \Gamma_3) <\L(\Gamma_2\Gamma_3) <\L(\Gamma_1 \Gamma_2 \Gamma_3)
\end{eqnarray*}
on the subsets of descriptions, which is exactly the aforementioned ordering $\L_1$. Fig.~\ref{fig:3AMD-setup} shows the setting of this problem for the ordering $\L_1$. It is worth mentioning that the distortion constraints may also induce different ordering of subsets of the descriptions. All possible ordering functions are listed in Table~\ref{tbl:ordering}. 

\begin{figure}[t]
\begin{center}
\psfrag{d1}[Bc][Bc]{$\ \mathtt{DEC}$}
\psfrag{d2}[Bc][Bc]{$\ \mathtt{DEC}$}
\psfrag{d3}[Bc][Bc]{$\ \mathtt{DEC}$}
\psfrag{d4}[Bc][Bc]{$\ \mathtt{DEC}$}
\psfrag{d5}[Bc][Bc]{$\ \mathtt{DEC}$}
\psfrag{d6}[Bc][Bc]{$\ \mathtt{DEC}$}
\psfrag{d7}[Bc][Bc]{$\ \mathtt{DEC}$}
 \psfrag{1}[Bc][Bc]{$\hspace{-15pt}(\hat{X}^n_{\Gamma_1},D_{\Gamma_1})$}
 \psfrag{2}[Bc][Bc]{$ \hspace{-15pt}(\hat{X}^n_{\Gamma_2},D_{\Gamma_2})$}
 \psfrag{3}[Bc][Bc]{$ \hspace{-15pt}(\hat{X}^n_{\Gamma_3},D_{\Gamma_3})$}
 \psfrag{4}[Bc][Bc]{$ \hspace{-1pt}(\hat{X}^n_{\Gamma_1 \Gamma_2},D_{\Gamma_1 \Gamma_2})$}
 \psfrag{5}[Bc][Bc]{$ \hspace{-1pt}(\hat{X}^n_{\Gamma_1 \Gamma_3},D_{\Gamma_1 \Gamma_3})$}
 \psfrag{6}[Bc][Bc]{$ \hspace{-1pt}(\hat{X}^n_{\Gamma_2 \Gamma_3},D_{\Gamma_2 \Gamma_3})$}
 \psfrag{7}[Bc][Bc]{$ \hspace{15pt}(\hat{X}^n_{\Gamma_1 \Gamma_2 \Gamma_3},D_{\Gamma_1 \Gamma_2 \Gamma_3})$}
 \psfrag{e}[Bc][Bc]{$\mathtt{ENC}$}
\psfrag{x}[Bc][Bc]{$X^n$}
\psfrag{a}[Bc][Bc]{$\Gamma_1$}
\psfrag{b}[Bc][Bc]{$\Gamma_2$}
\psfrag{c}[Bc][Bc]{$\Gamma_3$}
\hspace{2mm}
\includegraphics[width=10cm]{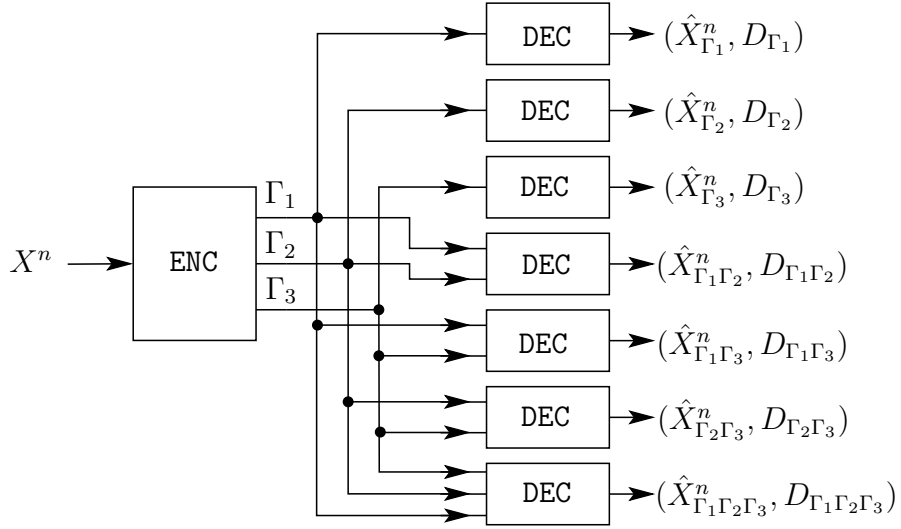}
\caption{The three description source coding problem with ordering level $\L_1$.}
\label{fig:3AMD-setup}
\end{center}
\end{figure}

\section{The Main Results}
In this section we present the main results of the paper. We state the theorems in a unified way which hold for all orderings, 
and also specialize it to the ordering $\L_1$ to facilitate understanding and further discussion. We start with the admissible rate region of the  A-MLD problem, $\R\MLD$, and then give an approximate characterization of the rate region of the  A-MD problem based on  the coding scheme inspired by the A-MLD problem.  


\subsection{The Admissible Rate Region of $3$-Description Asymmetric Multilevel Diversity Coding}
The following theorem characterizes the admissible rate region of the asymmetric multilevel diversity coding problem for an arbitrary ordering level.

\begin{theorem}
Let $\V=(V_1,\dots, V_7)$ be a given sequence of sources with entropy sequence $\mathbf{H}=(H_1,\dots,H_7)=(H(V_1),H(V_1,V_2),\dots,H(V_1,\dots,V_7))$. For a given ordering level $\L$, the rate region $\R\MLD^{\L}(\mathbf{H})$ is  the set of all  non-negative triples $(R_1,R_2,R_3)$ which satisfy

 {\makeatletter \newcounter{tempcnt}
  \setcounter{tempcnt}{\the\c@equation}
  \def\theequation{$P$\@arabic\c@equation}\makeatother
  \setcounter{equation}{0}
\begin{align}
R_i &\geq H_{\L(\Gamma_i)},\qquad i=1,2,3\label{ur1}\\
R_i+R_j &\geq H_{\min\{\L(\Gamma_i),\L(\Gamma_j)\}}+H_{\L(\Gamma_i,\Gamma_j)},\qquad i\neq j\label{ur12}\\
2R_i+R_j+R_k &\geq H_{\min\{\L(\Gamma_i),\L(\Gamma_j)\}}+H_{\min\{\L(\Gamma_i),\L(\Gamma_k)\}} \nonumber\\
&\phantom{=}+ H_{\min\{\L(\Gamma_i,\Gamma_j), \L(\Gamma_i,\Gamma_k)  \}}+H_{\L(\Gamma_i,\Gamma_j,\Gamma_k)}, \qquad i\neq j \neq k \label{ur1123}\\
R_1+R_2+R_3 &\geq H_{\L(\Gamma_1)} + H_{\min\{\L(\Gamma_1,\Gamma_2), \L(\Gamma_3) \}}   + H_{\L(\Gamma_1,\Gamma_2,\Gamma_3)},\label{ur123-1}\\
R_1+R_2+R_3 &\geq H_{\L(\Gamma_1)}+\frac{1}{2} H_{\L(\Gamma_2)} + \frac{1}{2} H_{\min\{\L(\Gamma_1,\Gamma_2)
, \L(\Gamma_1,\Gamma_3) ,\L(\Gamma_2,\Gamma_3) \}} \nonumber\\
&\phantom{=}+ H_{\L(\Gamma_1,\Gamma_2,\Gamma_3)}.\label{ur123-2}
\end{align}
\makeatletter
\setcounter{equation}{\the\c@tempcnt}\makeatother
}
\label{thm:mld}
\end{theorem}
 
In the following corollary, we specialize the bounds for  the specific ordering $\L_1$.

\begin{corollary}
For the ordering level $\L_1$, the admissible rate region of the three-description A-MLD problem is given by the set of all rate triples $(R_1,R_2,R_3)$ which satisfy
 {\makeatletter \newcounter{tempcnta}
  \setcounter{tempcnta}{\the\c@equation}
  \def\theequation{$Q$\@arabic\c@equation}\makeatother
  \setcounter{equation}{0}
\begin{align}
R_1 &\geq H(V_1),\label{r1}\\
R_2 &\geq H(V_1)+H(V_2),\label{r2}\\
R_3 &\geq H(V_1)+H(V_2)+H(V_3),\label{r3}\\
R_1+R_2 &\geq 2H(V_1)+H(V_2)+H(V_3) +H(V_4),\label{r12}\\
R_1+R_3 &\geq 2H(V_1)+H(V_2)+H(V_3) +H(V_4)+H(V_5),\label{r13}\\
R_2+R_3 &\geq 2H(V_1)+2H(V_2)+H(V_3) +H(V_4)+H(V_5)+H(V_6),\label{r23}\\
2R_1+R_2+R_3 &\geq 4H(V_1)+2H(V_2)+2H(V_3) +2H(V_4)+H(V_5)+H(V_6)+H(V_7),\label{r1123}\\
R_1+2R_2+R_3 &\geq 4H(V_1)+3H(V_2)+2H(V_3) +2H(V_4)+H(V_5)+H(V_6)+H(V_7),\label{r1223}\\
R_1+R_2+2R_3 &\geq 4H(V_1)+3H(V_2)+2H(V_3) +2H(V_4)+2H(V_5)+H(V_6)+H(V_7),\label{r1233}\\
R_1+R_2+R_3 &\geq   3H(V_1)+2H(V_2)+2H(V_3) +H(V_4)+H(V_5)+H(V_6)+H(V_7),\label{r123-1}\\
R_1+R_2+R_3 &\geq   3H(V_1)+2H(V_2)+\frac{3}{2}H(V_3) +\frac{3}{2}H(V_4)+H(V_5)+H(V_6)+H(V_7).
\label{r123-2}
\end{align}
\makeatletter
\setcounter{equation}{\the\c@tempcnt}\makeatother
}
\end{corollary}

\subsection{Approximate Rate Region Characterization of Gaussian Asymmetric $3$-Description Coding}
In the following theorems, we establish outer and inner bounds for the rate region of the Gaussian asymmetric multiple  descriptions coding. 
\begin{theorem}
For a given distortion vector $\D=(D_{\Gamma_1},\dots,D_{\Gamma_1\Gamma_2\Gamma_3})$, denote by $\uR\MD(\D)$ the set of all rate triples $(R_1,R_2,R_3)$ satisfying
{\makeatletter
\newcounter{tempcntd}
\setcounter{tempcntd}{\the\c@equation}
\def\theequation{$\mathcal{O}-$\@arabic\c@equation}\makeatother
\setcounter{equation}{0}
\begin{align}
R_i &\geq  \frac{1}{2} \log \frac{1}{D_{\Gamma_i}}, \qquad i=1,2,3 \label{uR1}\\
R_i+R_j & \geq \min\left(\frac{1}{2} \log \frac{1}{D_{\Gamma_i}}, \frac{1}{2} \log \frac{1}{D_{\Gamma_j}}\right)  +\frac{1}{2}\log\frac{1}{D_{\Gamma_i \Gamma_j} }-1, \qquad i\neq j \label{uR12}\\
2R_i+R_j +R_k & \geq \min\left(\frac{1}{2} \log \frac{1}{D_{\Gamma_i}}, \frac{1}{2} \log \frac{1}{D_{\Gamma_j}}\right)
+ \min\left(\frac{1}{2} \log \frac{1}{D_{\Gamma_i}}, \frac{1}{2} \log \frac{1}{D_{\Gamma_k}}\right)\nonumber\\
&\phantom{=}+\min\left(\frac{1}{2} \log \frac{1}{D_{\Gamma_i \Gamma_j}}, \frac{1}{2} \log \frac{1}{D_{\Gamma_i \Gamma_k}}\right)
+\frac{1}{2} \log \frac{1}{D_{\Gamma_i \Gamma_j \Gamma_k}}-3,\qquad i\neq j\neq k \label{uR1123} \\
R_1+R_2 +R_3 &\geq  \frac{1}{2} \log \frac{1}{D_{\Gamma_1} }  + \min\left(\frac{1}{2} \log\frac{1}{ D_{\Gamma_1 \Gamma_2}}, \frac{1}{2} \log\frac{1}{ D_{\Gamma_3}} \right)    \nonumber\\
&\phantom{=}  + \frac{1}{2} \log \frac{1}{D_{\Gamma_1 \Gamma_2 \Gamma_3} } -2, \label{uR123-2}\\
R_1+R_2 +R_3 &\geq \frac{1}{4} \log \frac{1}{D_{\Gamma_1}^2 D_{\Gamma_2}}  +  \min\left(\frac{1}{2}\log\frac{1}{ D_{\Gamma_1 \Gamma_2}}, \frac{1}{2}\log\frac{1}{ D_{\Gamma_1 \Gamma_3}}, \frac{1}{2}\log\frac{1}{ D_{\Gamma_2 \Gamma_3}}\right) \nonumber\\
&\phantom{=} + \frac{1}{2} \log \frac{1}{D_{\Gamma_1 \Gamma_2 \Gamma_3} }   -\frac{9}{2}.\label{uR123-1}
\end{align}
\makeatletter
\setcounter{equation}{\the\c@tempcnt}\makeatother
}
Then any admissible rate triple belongs to $\uR\MD(\D)$, \emph{i.e.,} $\R\MD(\D)\subseteq\uR\MD(\D)$. 
\label{thm:md-outer}
\end{theorem}

The bound stated in this theorem is a consequence of a more general parametric outer bound $\uR\MD^p(\D,\mathbf{d})$, defined in Theorem~\ref{thm:md-outer-parametric}. However, the current form is more convenient for comparison between the inner and outer bounds. This region is given in the following corollary for the specific ordering $\L_1$.
\begin{corollary}
Any admissible rate triple for a three-description A-MD with $\L_1$ ordering satisfies 
{\makeatletter
\newcounter{tempcnte}
\setcounter{tempcnte}{\the\c@equation}
\def\theequation{$\mathcal{O}'-$\@arabic\c@equation}\makeatother
\setcounter{equation}{0}
\begin{align}
R_1 &\geq  \frac{1}{2} \log \frac{1}{D_{\Gamma_1}},\\
R_2 &\geq  \frac{1}{2} \log \frac{1}{D_{\Gamma_2}},\\
R_3 &\geq  \frac{1}{2} \log \frac{1}{D_{\Gamma_3}},\\
R_1+R_2 & \geq \frac{1}{2} \log \frac{1}{D_{\Gamma_1} D_{\Gamma_1 \Gamma_2}}-1,\\
R_1+R_3 & \geq \frac{1}{2} \log \frac{1}{D_{\Gamma_1} D_{\Gamma_1 \Gamma_3}}-1,\\
R_2+R_3 & \geq \frac{1}{2} \log \frac{1}{D_{\Gamma_2} D_{\Gamma_2 \Gamma_3}}-1,\\
2R_1+R_2 +R_3 & \geq \frac{1}{2} \log \frac{1}{D_{\Gamma_1}^2 D_{\Gamma_1 \Gamma_2} D_{\Gamma_1 \Gamma_2 \Gamma_3}}-3,\\
R_1+2R_2 +R_3 & \geq \frac{1}{2} \log \frac{1}{D_{\Gamma_1} D_{\Gamma_2} D_{\Gamma_1 \Gamma_2} D_{\Gamma_1 \Gamma_2 \Gamma_3}}-3,\\
R_1+R_2 +2R_3 & \geq \frac{1}{2} \log \frac{1}{D_{\Gamma_1} D_{\Gamma_2} D_{\Gamma_1 \Gamma_3}  D_{\Gamma_1 \Gamma_2 \Gamma_3}}-3,\\
R_1+R_2 +R_3 & \geq \frac{1}{2} \log \frac{1}{D_{\Gamma_1} D_{\Gamma_3} D_{\Gamma_1 \Gamma_2 \Gamma_3}}-\frac{9}{2},\\
R_1+R_2+R_3 &\geq
\frac{1}{4} \log\frac{1}{D_{\Gamma_1}^2 D_{\Gamma_2} D_{\Gamma_1 \Gamma_2} D_{\Gamma_1 \Gamma_2 \Gamma_3}^2}-2.
\end{align}
\makeatletter
\setcounter{equation}{\the\c@tempcnt}\makeatother
}
\label{cor:md-outer}
\end{corollary}

Theorem~\ref{thm:md-inner} gives an inner bound for the admissible rate region of the three-description A-MD problem.

\begin{theorem}

For a given distortion vector $\D=(D_{\Gamma_1},\dots,D_{\Gamma_1 \Gamma_2 \Gamma_3})$, let $\oR\MD (\D)$ be the set of all rate triples $(R_1,R_2,R_3)$ satisfying

{\makeatletter \newcounter{tempcntb}
  \setcounter{tempcntb}{\the\c@equation}
  \def\theequation{$\mathcal{I}-$\@arabic\c@equation}\makeatother
  \setcounter{equation}{0}
\begin{align}
R_i &\geq  \frac{1}{2} \log \frac{1}{D_{\Gamma_i}}, \qquad i=1,2,3 \label{oR1}\\
R_i+R_j & \geq \min\left(\frac{1}{2} \log \frac{1}{D_{\Gamma_i}}, \frac{1}{2} \log \frac{1}{D_{\Gamma_j}}\right)  +\frac{1}{2}\log\frac{1}{D_{\Gamma_i \Gamma_j} }, \qquad i\neq j \label{oR12}\\
2R_i+R_j +R_k & \geq \min\left(\frac{1}{2} \log \frac{1}{D_{\Gamma_i}}, \frac{1}{2} \log \frac{1}{D_{\Gamma_j}}\right)
+ \min\left(\frac{1}{2} \log \frac{1}{D_{\Gamma_i}}, \frac{1}{2} \log \frac{1}{D_{\Gamma_k}}\right)\nonumber\\
&\phantom{=}+\min\left(\frac{1}{2} \log \frac{1}{D_{\Gamma_i \Gamma_j}}, \frac{1}{2} \log \frac{1}{D_{\Gamma_i \Gamma_k}}\right)
+\frac{1}{2} \log \frac{1}{D_{\Gamma_i \Gamma_j \Gamma_k}},\qquad i\neq j \neq k \label{oR1123} \\
R_1+R_2 +R_3 &\geq  \frac{1}{2} \log \frac{1}{D_{\Gamma_1} } +  \min\left(\frac{1}{2}\log\frac{1}{ D_{\Gamma_1 \Gamma_2}}, \frac{1}{2}\log\frac{1}{ D_{\Gamma_3}} \right) + \frac{1}{2} \log \frac{1}{D_{\Gamma_1 \Gamma_2 \Gamma_3} },   \label{oR123-2}\\
R_1+R_2 +R_3 &\geq \frac{1}{4} \log \frac{1}{D_{\Gamma_1}^2 D_{\Gamma_2}} +   \min\left(\frac{1}{2}\log\frac{1}{ D_{\Gamma_1 \Gamma_2}} ,\frac{1}{2}\log\frac{1}{ D_{\Gamma_1 \Gamma_3}}, \frac{1}{2}\log\frac{1}{ D_{\Gamma_2 \Gamma_3}}\right)       \nonumber\\
&\phantom{=} +\frac{1}{2} \log \frac{1}{D_{\Gamma_1 \Gamma_2 \Gamma_3} }. \label{oR123-1}
\end{align}
\makeatletter
\setcounter{equation}{\the\c@tempcnt}\makeatother}
Then any rate triple $\bR\in\oR\MD$ is achievable, \emph{i.e,} $\oR\MD\subseteq\R\MD$. 
\label{thm:md-inner}
\end{theorem} 

The following corollary specifies the above theorem for the ordering level $\L_1$. 

\begin{corollary}
If the distortion constraints satisfy the ordering level $\L_1$, \emph{i.e.,}
\begin{eqnarray*}
D_{\Gamma_1} \geq D_{\Gamma_2} \geq  D_{\Gamma_3} \geq D_{\Gamma_1 \Gamma_2} \geq D_{\Gamma_1 \Gamma_3} \geq D_{\Gamma_2 \Gamma_3} \geq D_{\Gamma_1 \Gamma_2 \Gamma_3},
\end{eqnarray*}
then any rate triple $(R_1,R_2,R_3)$ satisfying 
{\makeatletter \newcounter{tempcntc}
  \setcounter{tempcntc}{\the\c@equation}
  \def\theequation{$\mathcal{I}'-$\@arabic\c@equation}\makeatother
  \setcounter{equation}{0}
\begin{align}
R_1 &\geq  \frac{1}{2} \log \frac{1}{D_{\Gamma_1}},\\
R_2 &\geq  \frac{1}{2} \log \frac{1}{D_{\Gamma_2}},\\
R_3 &\geq  \frac{1}{2} \log \frac{1}{D_{\Gamma_3}},\\
R_1+R_2 & \geq \frac{1}{2} \log \frac{1}{D_{\Gamma_1} D_{\Gamma_1 \Gamma_2}},\\
R_1+R_3 & \geq \frac{1}{2} \log \frac{1}{D_{\Gamma_1} D_{\Gamma_1 \Gamma_3}},\\
R_2+R_3 & \geq \frac{1}{2} \log \frac{1}{D_{\Gamma_2} D_{\Gamma_2 \Gamma_3}},\\
2R_1+R_2 +R_3 & \geq \frac{1}{2} \log \frac{1}{D_{\Gamma_1}^2 D_{\Gamma_1 \Gamma_2} D_{\Gamma_1 \Gamma_2 \Gamma_3}},\\
R_1+2R_2 +R_3 & \geq \frac{1}{2} \log \frac{1}{D_{\Gamma_1} D_{\Gamma_2} D_{\Gamma_1 \Gamma_2} D_{\Gamma_1 \Gamma_2 \Gamma_3}},\\
R_1+R_2 +2R_3 & \geq \frac{1}{2}\log \frac{1}{D_{\Gamma_1} D_{\Gamma_3}^2 D_{\Gamma_1 \Gamma_2 \Gamma_3}},\\
R_1+R_2 +R_3 & \geq \frac{1}{2} \log \frac{1}{D_{\Gamma_1} D_{\Gamma_3} D_{\Gamma_1 \Gamma_2 \Gamma_3}},\\
R_1+R_2 +R_3 & \geq \frac{1}{4} \log\frac{1}{D_{\Gamma_1}^2 D_{\Gamma_2} D_{\Gamma_1 \Gamma_2} D_{\Gamma_1 \Gamma_2 \Gamma_3}^2},
\end{align}
\makeatletter
\setcounter{equation}{\the\c@tempcnt}\makeatother}
is achievable. 
\label{cor:md-inner}
\end{corollary} 

Summarizing the results of Theorem~\ref{thm:md-outer} and Theorem~\ref{thm:md-inner} gives the following corollary.

\begin{corollary}
\begin{eqnarray}
\overline{\R}\MD(\D)\subseteq \R\MD(\D) \subseteq \underline{\R}\MD(\D). 
\end{eqnarray}
\end{corollary}

The result of this corollary is that the multiple description admissible rate region is bounded between two sets of hyperplanes, which are pair-wise parallel. For each pair of parallel planes, we can compute the distance between them. Denote by $\delta_{(x,y,z)}$ the Euclidean distance between two parallel planes which are orthogonal to the vector $(x,y,z)$. Then for the distortion constraints corresponding to ordering $\L_1$, we have 
\begin{align}
\delta_{(1,0,0)}&= 0,\\
\delta_{(1,1,0)}&\leq \frac{1}{\sqrt{2}}=0.7071,\\
\delta_{(2,1,1)}&\leq \frac{3}{\sqrt{6}}=1.2247,\\
\delta_{(1,1,1)}&\leq \frac{9}{4\sqrt{3}}=1.2990,
\end{align}
where the denominators are the normalizing factors, corresponding to the length of the vector $(x,y,z)$. This shows that the inner and outer bounds provide an approximate characterization for the admissible rate region, for which the Euclidean distance between the bounds in less than  $1.3$ in the worst case. Fig.~\ref{region} shows a typical pair of  inner and outer bounds for $\L_1$ ordering and the case $D_{\Gamma_2}D_{\Gamma_1\Gamma_3} \leq D_{\Gamma_3}^2\leq D_{\Gamma_2}D_{\Gamma_1 \Gamma_2}$, which is the lossy counterpart of the lossless A-MLD problem with $h_4\leq h_3\leq h_4+h_5$, discussed in Subsection~\ref{subsec:MLD-ach}, under regime \mbox{II} (see also Fig.~\ref{reg2}). 

\begin{figure}[ht]
\begin{center}
	\psfrag{r1}[Bc][bc]{$R_1$}
	\psfrag{r2}[Bc][bc]{$R_2$}
	\psfrag{r3}[Bc][bc]{$R_3$}
	\psfrag{d1}[Bc][bc]{$\leq 0.7071$}
	\psfrag{d2}[Bc][bc]{$\leq 1.2247$}
	\psfrag{d3}[Bc][bc]{$\leq 1.299$}
	\includegraphics[width=10cm]{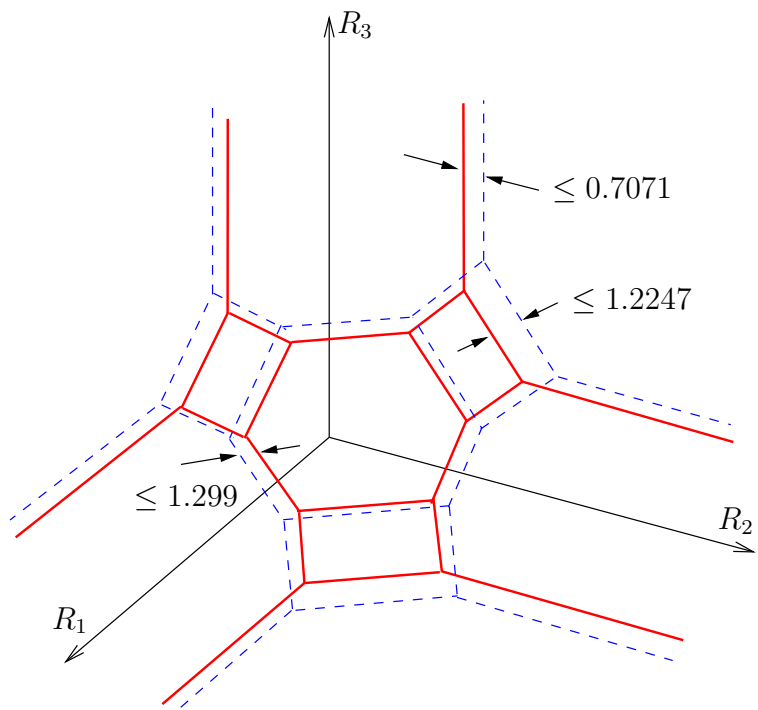}
\caption{The inner and outer bound for the admissible rate region of the Gaussian multiple descriptions problem for distortion constraints corresponding to the ordering $\L_1$, and the case $D_{\Gamma_2}D_{\Gamma_1\Gamma_3} \leq D_{\Gamma_3}^2\leq D_{\Gamma_2}D_{\Gamma_1 \Gamma_2}$.} \hspace{.1cm}
\label{region}
\end{center}
\end{figure}

\section{Asymmetric Multilevel Diversity Coding}
In this section we first prove the converse part of Theorem~\ref{thm:mld} for all orderings, and then the achievability part for ordering $\L_1$. Similar techniques can be used straightforwardly to prove the achievability for all the other orderings, and therefore complete the proof of Theorem~\ref{thm:mld}.

\subsection{The Converse Proof}
In this subsection we show that any admissible rate triple satisfies
(\ref{ur1})-(\ref{ur123-2}). The following important lemma, which 
simplifies the proof of the theorem, relates the entropy of the original
source to the reconstructed one. 
\begin{lemma}
Let $\S\subseteq\{\Gamma_1,\Gamma_2,\Gamma_3\}$ be a subset of descriptions available at a
decoder, and $i \leq j \leq \L(\S)$. Then
\begin{eqnarray}
H(\S | U_i^n) \geq H(\S | U_j^n) + n (H_j -H_i -\delta_n)
\end{eqnarray}
where $\delta_n \rightarrow 0$ as $n$ increases.
\label{lm:fano}
\end{lemma}

\begin{IEEEproof}
Note that $i \leq j \leq \L(\S)$. Therefore, the decoding requirement for the 
decoder with access to $\S$ implies that the reconstructed sequence 
$\hat{U}_j^n(\S)$ equals to $U_j^n$ with high probability. Then
\begin{align}
H(\S  | U_i^n)
&\stackrel{(a)}= H(\S, \hat{U}_j^n(\S) | U_i^n)\nonumber\\
&=
H(\S, U_j^n, \hat{U}_j^n(\S) | U_i^n)-H(U_j^n | \S, \hat{U}_j^n(\S) , U_i^n)\nonumber\\
& \geq 
H(\S , U_j^n| U_i^n ) - H (U_j^n | \hat{U}_j^n(\S))\nonumber\\
&\stackrel{(b)}=
H(\S | U_j^n) + H(U_j^n| U_i^n) - H (U_j^n | \hat{U}_j^n(\S)),
\label{eq:FanoLem}
\end{align}
where $(a)$ holds since $\hat{U}_j^n(\S)$ is function of $\S$,
for $j \leq \L(\S)$, and $(b)$ is due to the fact that $U_i^n$ is a subsequence of $U_j^n$ for $j\geq i$.
The underlying distribution of $U_i^n$ and $U_j^n$ implies $ H(U_j^n| U_i^n)=n(H_j-H_i)$. The last term in (\ref{eq:FanoLem}) can be upper bounded using the Fano's
inequality \cite{CoverThomas} as
\begin{eqnarray}
H (U_j^n | \hat{U}_j^n(\S))
\leq 
h_B(P_e) +P_e \log (|\mathcal{U}_j^n|-1)\label{Fano}
\leq 1+ n c P_e
\end{eqnarray}
where $P_e=\Pr(\hat{U}_j^n(\S) \neq U_j^n)< \e$, $h_B(p)$ defined as
\begin{align*}
h_B(p)=-p\log_2(p)-(1-p)\log_2(1-p)
\end{align*}
is the binary entropy function, and $c=\log|\mathcal{U}_j|$ is a constant. The proof is complete by setting $\delta_n=\frac{1}{n}+c P_e$.
\end{IEEEproof}

\begin{proof} [The converse proof of Theorem~\ref{thm:mld}]
Let $(R_1,R_2,R_3)$ be any admissible rate triple, and $\Gamma_i$ be a single
description with ordering level $\L(\Gamma_i)$. Recall $U^n_0=0$, and note that 
$\hat{U}_{\L(\Gamma_i)}^n(\Gamma_i)$ is a function of $\Gamma_i$. Thus
\begin{align}
n(R_i+\e)\geq H(\Gamma_i) = H(\Gamma_i | U_0^n) \stackrel{(\star)}\geq H(\Gamma_i| U^n_{\L(\Gamma_i)}
)+n( H_{\L(\Gamma_i)} -\delta_n)\geq n(H_{\L(\Gamma_i)}-\delta_n).
\end{align}
This proves
(\ref{ur1}). Note that here and in the rest of this proof all the inequalities labeled by $(\star)$ are due to Lemma~\ref{lm:fano}. 

Toward proving (\ref{ur12}), we can write
\begin{align}
n(R_i+R_j+2\e)
& \geq 
H(\Gamma_i) + H(\Gamma_j)\nonumber\\
& \geq 
H(\Gamma_i|U_0^n) + H(\Gamma_j|U_0^n)\nonumber\\
&\stackrel{(\star)}{\geq}
H(\Gamma_i|U_{\L(\Gamma_i)}^n) + n(H_{\L(\Gamma_i)}-\delta_n)+ H(\Gamma_j| U_{\L(\Gamma_j)}^n)+n(H_{\L(\Gamma_j)}-\delta_n)\nonumber\\
& \geq 
n (H_{\L(\Gamma_i)} +H_{\L(\Gamma_j)} -2\delta_n) +H(\Gamma_i|U_{\max\{\L(\Gamma_i),\L(\Gamma_j)\}}^n) +H(\Gamma_i|U_{\max\{\L(\Gamma_i),\L(\Gamma_j)\}}^n)\nonumber\\
&\geq
n (H_{\L(\Gamma_i)} +H_{\L(\Gamma_j)} -2\delta_n) +H(\Gamma_i,\Gamma_j|U_{\max\{\L(\Gamma_i),\L(\Gamma_j)\}}^n)\nonumber\\
& \stackrel{(\star)}{\geq} 
n( H_{\L(\Gamma_i)} +H_{\L(\Gamma_j)} -2\delta_n)+n(H_{\L(\Gamma_i,\Gamma_j)}-H_{\max\{\L(\Gamma_i),\L(\Gamma_j)\}}-\delta_n)\nonumber\\
&\phantom{=}+ H(\Gamma_i,\Gamma_j|U_{\L(\Gamma_i,\Gamma_j)}^n)\nonumber\\
&\geq 
n\left[H_{\min\{\L(\Gamma_i),\L(\Gamma_j)\}}+H_{\L(\Gamma_i,\Gamma_j)} -3\delta_n\right].
\end{align}

For proving (\ref{ur1123}) we can start with
\begin{align}
n(2R_i +R_j+&R_k+4\e)
\geq 2H(\Gamma_i)+H(\Gamma_j)+H(\Gamma_k)\nonumber\\
&\stackrel{(\star)}\geq 
[H(\Gamma_i|U_{\L(\Gamma_i)}^n) + nH_{\L(\Gamma_i)}-n\delta_n + H(\Gamma_j|U_{\L(\Gamma_j)}^n) + nH_{\L(\Gamma_j)} -n\delta_n]\nonumber\\
&\phantom{=}+[H(\Gamma_i|U_{\L(\Gamma_i)}^n) + nH_{\L(\Gamma_i)} -n\delta_n + H(\Gamma_k|U_{\L(\Gamma_k)}^n) + nH_{\L(\Gamma_k)} -n\delta_n]\nonumber\\
&\geq  n(2H_{\L(\Gamma_i)}+ H_{\L(\Gamma_j)} + H_{\L(\Gamma_k)} -4\delta_n) + H(\Gamma_i,\Gamma_j|U_{\max\{\L(\Gamma_i),\L(\Gamma_j)\}}^n) \nonumber\\
&\phantom{=}+ H(\Gamma_i,\Gamma_k|U_{\max\{\L(\Gamma_i),\L(\Gamma_k)\}}^n)\nonumber\\
&\stackrel{(\star)}\geq  n(2H_{\L(\Gamma_i)}+ H_{\L(\Gamma_j)} + H_{\L(\Gamma_k)} -4\delta_n) \nonumber\\
&\phantom{=}+ n(H_{\L(\Gamma_i,\Gamma_j)}-H_{\max\{\L(\Gamma_i),\L(\Gamma_j)\}} -\delta_n) + H(\Gamma_i,\Gamma_j|U_{\L(\Gamma_i,\Gamma_j)}^n)\nonumber\\
&\phantom{=}+n(H_{\L(\Gamma_i,\Gamma_k)}-H_{\max\{\L(\Gamma_i),\L(\Gamma_k)\}} -\delta_n) + H(\Gamma_i,\Gamma_k|U_{\L(\Gamma_i,\Gamma_k)}^n)\nonumber\\
&\geq  n(H_{\min\{\L(\Gamma_i),\L(\Gamma_j)\}}+ H_{\min\{\L(\Gamma_i),\L(\Gamma_k)\}}+ H_{\L(\Gamma_i,\Gamma_j)} + H_{\L(\Gamma_i,\Gamma_k)}-6\delta_n)\nonumber\\
&\phantom{=}+H(\Gamma_i,\Gamma_j,\Gamma_k|U_{\max\{ \L(\Gamma_i,\Gamma_j),\L(\Gamma_i,\Gamma_k) \}}^n)\nonumber\\
&\stackrel{(\star)}\geq
n(H_{\min\{\L(\Gamma_i),\L(\Gamma_j)\}}+ H_{\min\{\L(\Gamma_i),\L(\Gamma_k)\}}+ H_{\L(\Gamma_i,\Gamma_j)}+ H_{\L(\Gamma_i,\Gamma_k)} -6\delta_n)\nonumber\\
&\phantom{=}+n(H_{\L(\Gamma_i,\Gamma_j,\Gamma_k)}-H_{\max\{ \L(\Gamma_i,\Gamma_j),\L(\Gamma_i,\Gamma_k) \}}-\delta_n)\nonumber\\
&=n\left[H_{\min\{\L(\Gamma_i),\L(\Gamma_j)\}}+ H_{\min\{\L(\Gamma_i),\L(\Gamma_k)\}} \right.\nonumber\\
&\phantom{=}\left.+ H_{\min\{ \L(\Gamma_i,\Gamma_j),\L(\Gamma_i,\Gamma_k) \}}+ H_{\L(\Gamma_i,\Gamma_j,\Gamma_k)}-7\delta_n\right].
\end{align}

Toward proving \eqref{ur123-1} we can write

\begin{align}
n(R_1+R_2+R_3+3\e)&\geq H(\Gamma_1)+H(\Gamma_2)+H(\Gamma_3)\nonumber\\
&\stackrel{(\star)}{\geq}  n(H_{\L(\Gamma_1)}+H_{\L(\Gamma_2)}+H_{\min\{\L(\Gamma_1,\Gamma_2),\L(\Gamma_3)\}}-3\delta_n)\nonumber\\
&\phantom{=}+H(\Gamma_1|U_{\L(\Gamma_1)}^n)+H(\Gamma_2|U_{\L(\Gamma_2)}^n)+H(\Gamma_3|U_{\min\{\L(\Gamma_1,\Gamma_2),\L(\Gamma_3)\}}^n) \nonumber\\
& \geq  n(H_{\L(\Gamma_1)}+H_{\L(\Gamma_2)}+H_{\min\{\L(\Gamma_1,\Gamma_2),\L(\Gamma_3)\}}-3\delta_n)\nonumber\\
&\phantom{=}+H(\Gamma_1, \Gamma_2|U_{\L(\Gamma_2)}^n)+H(\Gamma_3|U_{\min\{\L(\Gamma_1,\Gamma_2),\L(\Gamma_3)\}}^n) \nonumber\\
&\stackrel{(\star)}{\geq}
  n(H_{\L(\Gamma_1)}+H_{\L(\Gamma_2)}+H_{\min\{\L(\Gamma_1,\Gamma_2),\L(\Gamma_3)\}}-3\delta_n)\nonumber\\
&\phantom{=}+H(\Gamma_1,\Gamma_2|U_{\min\{\L(\Gamma_1,\Gamma_2),\L(\Gamma_3)\}}^n)+n(H_{\min\{\L(\Gamma_1,\Gamma_2),\L(\Gamma_3)\}}-H_{\L(\Gamma_2)}-\delta_n)\nonumber\\
&\phantom{=} +H(\Gamma_3|U_{\min\{\L(\Gamma_1,\Gamma_2),\L(\Gamma_3)\}}^n) \nonumber\\
&\geq  n(H_{\L(\Gamma_1)}+2H_{\min\{\L(\Gamma_1,\Gamma_2),\L(\Gamma_3)\}}-4\delta_n)\nonumber\\
&\phantom{=}+H(\Gamma_1,\Gamma_2,\Gamma_3 | U_{\min\{\L(\Gamma_1,\Gamma_2),\L(\Gamma_3)\}}^n)\nonumber\\
&\stackrel{(\star)}{\geq}  n(H_{\L(\Gamma_1)}+2H_{\min\{\L(\Gamma_1,\Gamma_2),\L(\Gamma_3)\}}-4\delta_n)\nonumber\\
&\phantom{=}+n(H_{\L(\Gamma_1,\Gamma_2,\Gamma_3)}-H_{\min\{\L(\Gamma_1,\Gamma_2),\L(\Gamma_3)\}}-\delta_n)\nonumber\\
&\geq  n\left[H_{\L(\Gamma_1)}+H_{\min\{\L(\Gamma_1,\Gamma_2),\L(\Gamma_3)\}}+H_{\L(\Gamma_1,\Gamma_2,\Gamma_3)}-5\delta_n\right].
\label{sum2}
\end{align}

We need to consider two different cases in order to obtain the other sum-rate  bound in \eqref{ur123-2}. First consider the case $\L(\Gamma_3)> \L(\Gamma_1,\Gamma_2)$. Note that this implies $\min\{ \L(\Gamma_1, \Gamma_2), \L(\Gamma_1, \Gamma_3) , \L(\Gamma_2, \Gamma_3)\} =\L(\Gamma_1, \Gamma_2)$. We have
\begin{align}
n(R_1+&R_2+R_3+3\e)\geq H(\Gamma_1)+H(\Gamma_2)+H(\Gamma_3)\nonumber\\ 
&\stackrel{(\star)}{\geq} 
n(H_{\L(\Gamma_1)}+H_{\L(\Gamma_2)}+H_{\L(\Gamma_3)}-3\delta_n)+H(\Gamma_1|U_{\L(\Gamma_1)}^n)+H(\Gamma_2|U_{\L(\Gamma_2)}^n)+H(\Gamma_3|U_{\L(\Gamma_3)}^n)\nonumber
\\ 
&=
n(H_{\L(\Gamma_1)}+ H_{\L(\Gamma_2)}+H_{\L(\Gamma_3)}-3\delta_n)+\frac{1}{2}[H(\Gamma_1|U_{\L(\Gamma_1)}^n)+H(\Gamma_2|U_{\L(\Gamma_2)}^n)]\nonumber\\
&
\qquad+\frac{1}{2}[H(\Gamma_1|U_{\L(\Gamma_1)}^n)+H(\Gamma_3|U_{\L(\Gamma_3)}^n)]+\frac{1}{2}[H(\Gamma_2|U_{\L(\Gamma_2)}^n)+H(\Gamma_3|U_{\L(\Gamma_3)}^n)]\nonumber\\
&\geq  n(H_{\L(\Gamma_1)}+H_{\L(\Gamma_2)}+H_{\L(\Gamma_3)}-3\delta_n)\nonumber\\
&\phantom{=}+\frac{1}{2}[H(\Gamma_1,\Gamma_2|U_{\L(\Gamma_2)}^n)+H(\Gamma_1,\Gamma_3|U_{\L(\Gamma_3)}^n)+H(\Gamma_2,\Gamma_3|U_{\L(\Gamma_3)}^n)] \label{eq:prf:mld:sum-1}\\ 
&\stackrel{(\star)}{\geq}
n(H_{\L(\Gamma_1)}+H_{\L(\Gamma_2)}+H_{\L(\Gamma_3)}-3\delta_n)\nonumber\\
&\phantom{=}+\frac{1}{2}\left[H(\Gamma_1,\Gamma_2|U_{\L(\Gamma_1, \Gamma_2)}^n)+n(H_{\L(\Gamma_1, \Gamma_2)}-H_{\L(\Gamma_2)} -\delta_n) \right]\nonumber\\
&\phantom{=}+ \frac{1}{2}\left[H(\Gamma_1,\Gamma_3|U_{\L(\Gamma_3)}^n)+H(\Gamma_2,\Gamma_3|U_{\L(\Gamma_3)}^n)\right]\nonumber\\
&\stackrel{(a)}{\geq}
n(H_{\L(\Gamma_1)}+\frac{1}{2}H_{\L(\Gamma_2)}+\frac{1}{2}H_{\L(\Gamma_1, \Gamma_2)}+H_{\L(\Gamma_3)}-\frac{7}{2}\delta_n)\nonumber\\
&\phantom{=}+\frac{1}{2}\left[H(\Gamma_1,\Gamma_2|U_{\L(\Gamma_3)}^n) + H(\Gamma_1,\Gamma_3|U_{\L(\Gamma_3)}^n)+H(\Gamma_2,\Gamma_3|U_{\L(\Gamma_3)}^n)\right]\nonumber\\
&\stackrel{(b)}{\geq} 
n(H_{\L(\Gamma_1)}+\frac{1}{2}H_{\L(\Gamma_2)}+\frac{1}{2}H_{\L(\Gamma_1, \Gamma_2)}+H_{\L(\Gamma_3)}-\frac{7}{2}\delta_n)+H(\Gamma_1,\Gamma_2,\Gamma_3|U_{\L(\Gamma_3)}^n)\nonumber\\
&\stackrel{(\star)}{\geq} 
n(H_{\L(\Gamma_1)}+\frac{1}{2}H_{\L(\Gamma_2)}+\frac{1}{2}H_{\L(\Gamma_1, \Gamma_2)}+H_{\L(\Gamma_3)}-\frac{7}{2}\delta_n)+n(H_{\L(\Gamma_1,\Gamma_2,\Gamma_3)}-H_{\L(\Gamma_3)}-\delta_n)\nonumber\\
&= n\left[H_{\L(\Gamma_1)}+\frac{1}{2}H_{\L(\Gamma_2)}+\frac{1}{2}H_{\L(\Gamma_1, \Gamma_2)}+H_{\L(\Gamma_1,\Gamma_2,\Gamma_3)}-\frac{9}{2}\delta_n\right],
\label{sum1-1}
\end{align}
where in  $(a)$ we have used $H(\Gamma_1,\Gamma_2|U_{\L(\Gamma_1,\Gamma_2)}^n)  \geq  H(\Gamma_1,\Gamma_2|U_{\L(\Gamma_3)}^n)$, implied by the assumption $\L(\Gamma_3)> \L(\Gamma_1,\Gamma_2)$, and $(b)$ is due to the
conditional version of Han's inequality \cite[page 491]{CoverThomas}. 
For the second case, \emph{i.e.,}
$\L(\Gamma_3) < \L(\Gamma_1,\Gamma_2)$, we have from \eqref{eq:prf:mld:sum-1},
\begin{align}
n(R_1+&R_2+R_3+3\e)
\geq  n(H_{\L(\Gamma_1)}+H_{\L(\Gamma_2)}+H_{\L(\Gamma_3)}-3\delta_n)\nonumber\\
&\phantom{=}+\frac{1}{2}\left[H(\Gamma_1,\Gamma_2|U_{\L(\Gamma_2)}^n)+H(\Gamma_1,\Gamma_3|U_{\L(\Gamma_3)}^n)+H(\Gamma_2,\Gamma_3|U_{\L(\Gamma_3)}^n)\right] \nonumber\\ 
&\stackrel{(\star)}{\geq}
n(H_{\L(\Gamma_1)}+H_{\L(\Gamma_2)}+H_{\L(\Gamma_3)}-3\delta_n)+\frac{1}{2}\Big[H(\Gamma_1,\Gamma_2|U_{\min\{\L(\Gamma_1,\Gamma_2),\L(\Gamma_1,\Gamma_3),\L(\Gamma_2,\Gamma_3)\}}^n)\nonumber\\
&\phantom{=}+ H(\Gamma_1,\Gamma_3|U_{\min\{\L(\Gamma_1,\Gamma_2),\L(\Gamma_1,\Gamma_3),\L(\Gamma_2,\Gamma_3)\}}^n)+H(\Gamma_2,\Gamma_3|U_{\min\{\L(\Gamma_1,\Gamma_2),\L(\Gamma_1,\Gamma_3),\L(\Gamma_2,\Gamma_3)\}}^n)\nonumber\\
&\phantom{=}+n(3 H_{\min\{\L(\Gamma_1,\Gamma_2),\L(\Gamma_1,\Gamma_3),\L(\Gamma_2,\Gamma_3)\}} 
-H_{\L(\Gamma_2)} -2H_{\L(\Gamma_3)} -3\delta_n)\Big]\nonumber\\
&\stackrel{(c)}{\geq} 
n(H_{\L(\Gamma_1)}+\frac{1}{2}H_{\L(\Gamma_2)}+\frac{3}{2}H_{\min\{\L(\Gamma_1,\Gamma_2),\L(\Gamma_1,\Gamma_3),\L(\Gamma_2,\Gamma_3)\}})-\frac{9}{2}\delta_n)\nonumber\\
&\phantom{=}+H(\Gamma_1,\Gamma_2,\Gamma_3|U_{\min\{\L(\Gamma_1,\Gamma_2),\L(\Gamma_1,\Gamma_3),\L(\Gamma_2,\Gamma_3)\}}^n)\nonumber\\
&\geq 
n(H_{\L(\Gamma_1)}+\frac{1}{2}H_{\L(\Gamma_2)}+\frac{3}{2}H_{\min\{\L(\Gamma_1,\Gamma_2),\L(\Gamma_1,\Gamma_3),\L(\Gamma_2,\Gamma_3)\}}-\frac{9}{2}\delta_n)\nonumber\\
&\phantom{=}+n(H_{\L(\Gamma_1,\Gamma_2,\Gamma_3)}-H_{\min\{\L(\Gamma_1,\Gamma_2),\L(\Gamma_1,\Gamma_3),\L(\Gamma_2,\Gamma_3)\}}-\delta_n)\nonumber\\
&= n\left[H_{\L(\Gamma_1)}+\frac{1}{2}H_{\L(\Gamma_2)}+\frac{1}{2}H_{\min\{\L(\Gamma_1,\Gamma_2),\L(\Gamma_1,\Gamma_3),\L(\Gamma_2,\Gamma_3)\}}+H_{\L(\Gamma_1,\Gamma_2,\Gamma_3)}-\frac{11}{2}\delta_n\right].
\label{sum1-2}
\end{align}
Again we have used the conditional Han's inequality in $(c)$. Putting \eqref{sum1-1} and \eqref{sum1-2} together, we obtain the bound \eqref{ur123-2}. 
\end{proof}

\subsection{Achievability}
\label{subsec:MLD-ach}
In the following we will show that the inequalities
(\ref{ur1})--(\ref{ur123-2}) provide a complete characterization of the
achievable rate region of the A-MLD problem. However, each individual case 
given in Table~\ref{tbl:ordering} needs to be considered 
separately, due to the specific strategy used in the coding scheme. For conciseness, 
we only present the analysis for the ordering level  $\L_1$, and provide 
the details of the achievability scheme for this specific ordering. 
More precisely, we show that any rate
triple $(R_1,R_2,R_3)$ satisfying \eqref{r1}--\eqref{r123-2}
is achievable, \emph{i.e.,} there exist encoding and decoding
functions with the desired rates which are able to reconstruct the required 
subset of the sources from the corresponding descriptions. This implies
$\R^{\L}\MLD(\mathbf{H})$ is achievable, and completes the proof of the theorem
for the ordering $\L_1$. Similar proof for other orderings can be straightforwardly 
completed by applying almost identical techniques. Different cases that needed to be considered 
are listed in Table~\ref{tbl:ordering}.

\begin{table}
\caption{The eight possible level orderings and the corresponding sub-regimes. }
\label{tbl:ordering}
\begin{center}
\begin{tabular}{|c|c|}
\hline
\textrm{Ordering} & \textrm{Regime}\\
\hline
\hline
\multirow{3}{*}{
$\L(\Gamma_1) < \L(\Gamma_2)< \L(\Gamma_3) < \L(\Gamma_{12}) < \L(\Gamma_{13}) < \L(\Gamma_{23}) < \L(\Gamma_{123})$} & 
$h_3\leq h_4$\\
\cline{2-2}
& $h_4\leq h_3 \leq h_4+h_5$\\
\cline{2-2}
& $h_3\geq h_4+h_5$
\\
\hline
\multirow{3}{*}{
$\L(\Gamma_1) < \L(\Gamma_2)< \L(\Gamma_3) < \L(\Gamma_{12}) < \L(\Gamma_{23}) < \L(\Gamma_{13}) < \L(\Gamma_{123})$} & 
$h_3\leq h_4$\\
\cline{2-2}
&$h_4\leq h_3 \leq h_4+h_5$\\
\cline{2-2}
&$h_3\geq h_4+h_5$
\\
\hline
\multirow{2}{*}{
$\L(\Gamma_1) < \L(\Gamma_2)< \L(\Gamma_3) < \L(\Gamma_{13}) < \L(\Gamma_{12}) < \L(\Gamma_{23}) < \L(\Gamma_{123})$} & 
$h_3\leq h_4$\\
\cline{2-2}
& $h_3 \geq h_4$ 
\\
\hline
\multirow{2}{*}{
$\L(\Gamma_1) < \L(\Gamma_2)< \L(\Gamma_3) < \L(\Gamma_{13}) < \L(\Gamma_{23}) < \L(\Gamma_{12}) < \L(\Gamma_{123})$} & 
$h_3\leq h_4$\\
\cline{2-2}
&$h_3\geq h_4$
\\
\hline
\multirow{2}{*}{
$\L(\Gamma_1) < \L(\Gamma_2)< \L(\Gamma_3) < \L(\Gamma_{23}) < \L(\Gamma_{12}) < \L(\Gamma_{13}) < \L(\Gamma_{123})$} & 
$h_3\leq h_4$\\
\cline{2-2}
&$h_3\geq h_4$
\\
\hline
\multirow{2}{*}{
$\L(\Gamma_1) < \L(\Gamma_2)< \L(\Gamma_3) < \L(\Gamma_{23}) < \L(\Gamma_{13}) < \L(\Gamma_{12}) < \L(\Gamma_{123})$} & 
$h_3\leq h_4$\\
\cline{2-2}
&$h_3\geq h_4$
\\
\hline
\multirow{2}{*}{
$\L(\Gamma_1) < \L(\Gamma_2)< \L(\Gamma_{12}) < \L(\Gamma_{3}) < \L(\Gamma_{13}) < \L(\Gamma_{23}) < \L(\Gamma_{123})$} & 
$h_3\leq h_5$\\
\cline{2-2}
& $h_3\geq h_5$ 
\\
\hline
\multirow{2}{*}{
$\L(\Gamma_1) < \L(\Gamma_2)< \L(\Gamma_{12}) < \L(\Gamma_{3}) < \L(\Gamma_{23}) < \L(\Gamma_{13}) < \L(\Gamma_{123})$} & 
$h_3\leq h_5$\\
\cline{2-2}
&$h_3\geq h_5$
\\
\hline
\end{tabular} 
\end{center}
\end{table}

Note that the $\R^{\L}\MLD$ is a polytopes specified by several 
hyperplanes in a three-dimensional space. Therefore, the region 
$\R^{\L}\MLD$ is a \emph{convex} polytopes, and it suffices to show the 
achievability only
for the \emph{corner points} 
\cite{Gold:56}; that is because a simple \emph{time-sharing} argument can be used  to 
 extend the achievability to any arbitrary point in the region
$\R^{\L}\MLD$.

Depending on the relationship of $h_3$, $h_4$, and $h_5$, some of the
inequalities in (\ref{r1})--(\ref{r123-2}) may be dominated by the
others. Note that (\ref{r123-1}) and (\ref{r123-2}) are of the form 
\begin{align*}
R_1+R_2+R_3 &\geq   3H(V_1)+2H(V_2)+2H(V_3) +H(V_4)+H(V_5)+H(V_6)+H(V_7),\\
R_1+R_2+R_3 &\geq   3H(V_1)+2H(V_2)+\frac{3}{2}H(V_3) +\frac{3}{2}H(V_4)+H(V_5)+H(V_6)+H(V_7).
\end{align*}
It is clear either one of them would be redundant and implied by the other, 
depending on whether $h_3\lessgtr h_4$.  Also if $h_3\geq h_4+h_5$, inequalities (\ref{r3}) and (\ref{r123-1}) imply
\begin{align*}
R_1+R_2+2R_3&\geq H_1+H_3+H_7+H_3\\
&\geq  H_1+H_3+H_7+H_2+h_4+h_5\\
&= H_1+H_2+H_5+H_7,
\end{align*}
which is exactly the inequality given in (\ref{r1233}), \emph{i.e.,} this inequality is  redundant in this regime. Thus, we  split the achievability proof  into  three 
regimes corresponding to the aforementioned conditions, since the proposed encoding 
schemes are slightly different for these regimes. We show the achievability
of the corner points in each case.

%

To simplify matters, we perform a lossless \emph{pre-coding}, acting
on all the seven source sequences $V_i^n$'s as
\begin{align*}
E_i: \mathcal{V}_i^n&\longrightarrow \{0,1\}^{\ell_i}
\end{align*}
for $i=1,\dots,7$. This function maps the source sequence $V_i^n$ to 
$\tilde{V_i}\triangleq E_i(V_i^n)$, which can be used as a new binary 
source sequence of length $\ell_i$. This can be done by using any lossless scheme, and achieves
$\ell_i$ arbitrary close to $ nh_i$ for large enough $n$. With the new 
source sequences $\tilde{V}_i$,  we next perform further coding.

\textbf{Regime I: $h_3\geq h_4+h_5$}\\
As mentioned above, the inequalities \eqref{r1233} and \eqref{r123-2}  are dominated by the others in this regime. Therefore we only need to consider the remaining nine
hyperplanes. In the following we list the corner points of
$\R^{\L}\MLD$ in this regime. Each corner point with coordinates $(R_1,R_2,R_3)$ is the intersection of
(at least) three hyperplane, say ($Qi$), ($Qj$),  and ($Qk$).  Such point is denoted 
by $\langle Qi,Qj,Qk\rangle:(R_1,R_2,R_3)$. In order to list all the corner points, we first find the intersection of any three hyperplanes, and then check whether the intersection point satisfies all the other inequalities. We next provide an
encoding strategy to achieve the rates prescribed by the corner points of the polytope. 
\begin{itemize}
\item $X_1=\langle$\ref{r1},\ref{r12},\ref{r1123}$\rangle: (H_1,H_4,H_7)$

This corner point is the intersection of the planes $Q_1$, $Q_4$, and $Q_7$, and determines the individual
rates of the descriptions as
\begin{eqnarray*}
(R_1,R_2,R_3)=(H_1,H_4,H_7).
\end{eqnarray*}
The scheme  for achieving this rate tuple is as
follows. $\Gamma_1$ is exactly the pre-coded sequence of $V_1^n$,
\emph{i.e.,} $\V_1$. In order to construct $\Gamma_2$ it suffices to 
concatenate the codewords $\V_1$, $\V_2$, $\V_3$, and
$\V_4$. Similarly, $\Gamma_3$ is the concatenation of all the seven
codewords. That is,
\begin{eqnarray*}
\Gamma_1:\V_1, \qquad \Gamma_2:\V_1,\V_2,\V_3,\V_4, \qquad \Gamma_3:\V_1,\V_2,\V_3,\V_4,\V_5,\V_6,\V_7.
\end{eqnarray*}
It is easy to check that the description rates are the same as the rate triple of the corner point, and all the decoding requirements at the seven decoders are satisfied. 

We will only determine the rate triples and illustrate the descriptions construction for the remaining corner points.
 
\item $X_2=\langle$\ref{r1},\ref{r13},\ref{r1123}$\rangle: (H_1,H_7-h_5,H_5)$
\begin{eqnarray*}
\Gamma_1:\V_1, \qquad \Gamma_2:\V_1,\V_2,\V_3,\V_4,\V_6,\V_7, \qquad \Gamma_3:\V_1,\V_2,\V_3,\V_4,\V_5.
\end{eqnarray*}

\item $X_3=\langle$\ref{r2},\ref{r12},\ref{r1223}$\rangle:(H_1+h_3+h_4,H_2, H_7 )$
\begin{eqnarray*}
\Gamma_1:\V_1,\V_3,\V_4, \qquad \Gamma_2:\V_1,\V_2, \qquad \Gamma_3:\V_1,\V_2,\V_3,\V_4,\V_5,\V_6,\V_7.
\end{eqnarray*}

\item $X_4=\langle$\ref{r2},\ref{r23},\ref{r1223}$\rangle:(H_1+h_3+h_4+h_7,H_2,H_6 )$
\begin{eqnarray*}
\Gamma_1:\V_1,\V_3,\V_4,\V_7, \qquad \Gamma_2:\V_1,\V_2, \qquad \Gamma_3:\V_1,\V_2,\V_3,\V_4,\V_5,\V_6.
\end{eqnarray*}

\item $X_5=\langle$\ref{r3},\ref{r13},\ref{r123-1}$\rangle:(H_1+h_4+h_5,H_3+h_6+h_7, H_3)$

The encoding schemes for the previous corner points only involve concatenation 
of different codewords. However, concatenation is not optimal to achieve the rate triple
induced by the point $X_5$,  and we need to jointly encode the sources to construct 
the descriptions. This can be done using a modulo-$2$ summation of (parts of) the 
codewords of the same size.

The description $\Gamma_1$ is simply constructed by concatenating  
$\V_1$, $\V_4$, and $\V_5$. Similarly, $\Gamma_3$ is obtained by putting 
$\V_1$, $\V_2$, and $\V_3$ together. The second description, $\Gamma_2$, should be able 
to help $\Gamma_1$ to reconstruct $\V_3$ at the decoder with access to $\{\Gamma_1,\Gamma_2\}$, 
and help $\Gamma_3$ to reconstruct $(\V_4,\V_5)$ at decoder
$\{\Gamma_2,\Gamma_3\}$, where $\V_3$ is already provided as a part of $\Gamma_3$. 
We can use this fact to construct $\Gamma_2$ as follows. Partition\footnote{Since we are in regime I, we have $h_3\geq h_4+h_5$ and
hence, $\ell_3\geq \ell_4+\ell_5$.} 
the bit stream $\V_3$ into $\V_{3,1}$ and $\V_{3,2}$ of
lengths $\ell_3-(\ell_4+\ell_5)$ and $\ell_4+\ell_5$, respectively. 
Compute the modulo-$2$ summation  (binary \texttt{xor}) of the bitstreams $\V_{3,2}$ and 
$(\V_4,\V_5)$. The description $\Gamma_2$ is constructed by concatenating this 
new bit stream with $\V_1$, $\V_2$, $\V_{3,1}$, $\V_6$, and $\V_7$.

The partitioning 
and encoding\footnote{Note that for this corner point, a part 
of the description $\Gamma_2$ is given by $\V_{3,2}\oplus(\V_4,\V_5)$,
which linearly combines independent (compressed) source sequences,
just as the network  coding idea in the familiar Butterfly network \cite{ACLY:00}.} are illustrated in Fig.~\ref{fig:ex-part}.
\begin{eqnarray*}
\Gamma_1:\V_1,\V_4,\V_5, \qquad
\Gamma_2:\V_1,\V_2,\V_{3,1},\V_{3,2}\oplus(\V_4,\V_5),\V_6,\V_7, \qquad
\Gamma_3:\V_1,\V_2,\V_3.
\end{eqnarray*}

\begin{figure}[ht]
\begin{center}
	\psfrag{U1}[Bc][Bc]{$\tilde{V}_1$}
	\psfrag{U2}[Bc][bc]{$\tilde{V}_2$}
	\psfrag{U3}[Bc][bc]{$\tilde{V}_{3,1}$}
	\psfrag{U32}[Bc][bc]{$\tilde{V}_{3,2}$}
	\psfrag{U4}[Bc][bc]{$\tilde{V}_4$}
	\psfrag{U5}[Bc][bc]{$\tilde{V}_{5}$}
	\psfrag{U6}[Bc][bc]{$\tilde{V}_6$}
	\psfrag{U7}[Bc][bc]{$\tilde{V}_7$}
	\psfrag{l1}[Bc][bc]{$\ell_1$}
	\psfrag{l2}[Bc][bc]{$\ell_2$}
	\psfrag{l31}[Bc][bc]{$\ell_3$}
	\psfrag{l33}[Bc][bc]{$-(\ell_4+\ell_5)$}	
	\psfrag{l32}[Bc][bc]{$\ell_4+\ell_5$}
	\psfrag{l4}[Bc][bc]{$\ell_4$}
	\psfrag{l5}[Bc][bc]{$\ell_5$}
	\psfrag{l6}[Bc][bc]{$\ell_6$}
	\psfrag{l7}[Bc][bc]{$\ell_7$}
	\psfrag{g1:}[Bc][bc]{$\Gamma_1:\ $}
	\psfrag{g2:}[Bc][bc]{$\Gamma_2:\ $}
	\psfrag{g3:}[Bc][bc]{$\Gamma_3:\ $}
	\includegraphics[width=150mm]{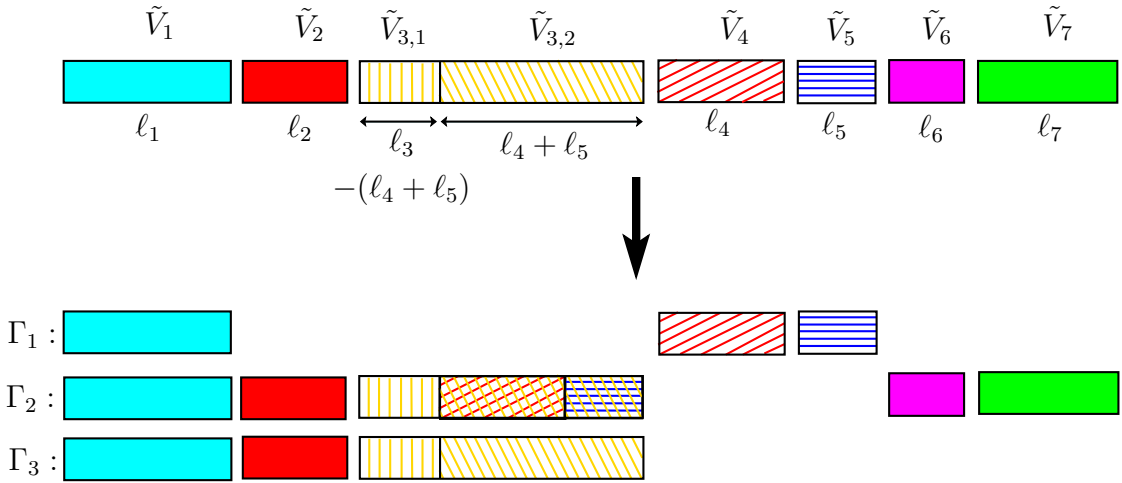}
\end{center}
\caption{Linear encoding for the corner-point $X_{5}$} 
\label{fig:ex-part}
\end{figure}

\item $X_6=\langle$\ref{r3},\ref{r23},\ref{r123-1}$\rangle:(H_1+h_3+h_7,H_2+h_4+h_5+h_6, H_3)$

Partition $\V_3$ into $\V_{3,1}$ and $\V_{3,2}$ of lengths
$\ell_3-(\ell_4+\ell_5)$ and $\ell_4+\ell_5$, respectively.
\begin{eqnarray*}
\Gamma_1:\V_1,\V_{3,1},\V_{3,2}\oplus(\V_4,\V_5),\V_7 \qquad
\Gamma_2:\V_1,\V_2,\V_4,\V_5,\V_6, \qquad \Gamma_3:\V_1,\V_2,\V_3.
\end{eqnarray*}

\item $X_7=\langle$\ref{r12},\ref{r1123},\ref{r123-1}$\rangle:(H_1+h_4, H_3, H_3+h_5+h_6+h_7 )$

Partition $\V_3$ into $\V_{3,1}$ and $\V_{3,2}$ of lengths $\ell_3-\ell_4$ and $\ell_4$, respectively.
\begin{eqnarray*}
\Gamma_1:\V_1,\V_4, \qquad \Gamma_2:\V_1,\V_2,\V_{3,1},\V_{3,2}\oplus\V_4, \qquad \Gamma_3:\V_1,\V_2,\V_3,\V_5,\V_6,\V_7.
\end{eqnarray*}

\item $X_8=\langle$\ref{r12}, \ref{r1223}, \ref{r123-1}$\rangle:( H_1+h_3, H_2+h_4, H_3+h_5+h_6+h_7 )$

Partition $\V_3$ into $\V_{3,1}$ and $\V_{3,2}$ of lengths $\ell_3-\ell_4$ and $\ell_4$, respectively.
\begin{eqnarray*}
\Gamma_1:\V_1,\V_{3,1},\V_{3,2}\oplus\V_4, \qquad \Gamma_2:\V_1,\V_2,\V_4, \qquad \Gamma_3:\V_1,\V_2,\V_3,\V_5,\V_6,\V_7.
\end{eqnarray*}

\item $X_9=\langle$\ref{r13}, \ref{r1123}, \ref{r123-1}$\rangle:(H_1+h_4, H_3+h_6+h_7, H_3+h_5 )$

Partition $\V_3$ into $\V_{3,1}$ and $\V_{3,2}$ of lengths $\ell_3-\ell_4$ and $\ell_4$, respectively.
\begin{eqnarray*}
\Gamma_1:\V_1,\V_4,  \qquad \Gamma_2:\V_1,\V_2,\V_{3,1},\V_{3,2}\oplus\V_4,\V_6,\V_7, \qquad \Gamma_3:\V_1,\V_2,\V_3,\V_5.
\end{eqnarray*}

\item $X_{10}=\langle$\ref{r23}, \ref{r1223}, \ref{r123-1}$\rangle:(H_1+h_3+h_7, H_2+h_4, H_3+h_5+h_6 )$

Partition $\V_3$ into $\V_{3,1}$ and $\V_{3,2}$ of lengths $\ell_3-\ell_4$ and $\ell_4$, respectively.
\begin{eqnarray*}
\Gamma_1:\V_1,\V_{3,1},\V_{3,2}\oplus\V_4,\V_7,  \qquad \Gamma_2:\V_1,\V_2,\V_4, \qquad \Gamma_3:\V_1,\V_2,\V_3,\V_5,\V_6.
\end{eqnarray*}
\end{itemize}
The associated rate region is shown in Fig.~\ref{reg1}.
\begin{figure}[ht]
\begin{center}
   \psfrag{x1}[Bc][Bc]{$X_1$}
   \psfrag{x2}[Bc][bc]{$X_2$}
   \psfrag{x3}[Bc][bc]{$X_3$}
   \psfrag{x4}[Bc][bc]{$X_4$}
   \psfrag{x5}[Bc][bc]{$X_5$}
   \psfrag{x6}[Bc][bc]{$X_6$}
   \psfrag{x7}[Bc][bc]{$X_7$}
   \psfrag{x8}[Bc][bc]{$X_8$}
   \psfrag{x9}[Bc][bc]{$X_9$}
   \psfrag{x10}[Bc][bc]{$X_{10}$}
   \psfrag{r1}[Bc][bc]{$R_1$}
   \psfrag{r2}[Bc][bc]{$R_2$}
   \psfrag{r3}[Bc][bc]{$R_3$}
   \includegraphics[width=11cm]{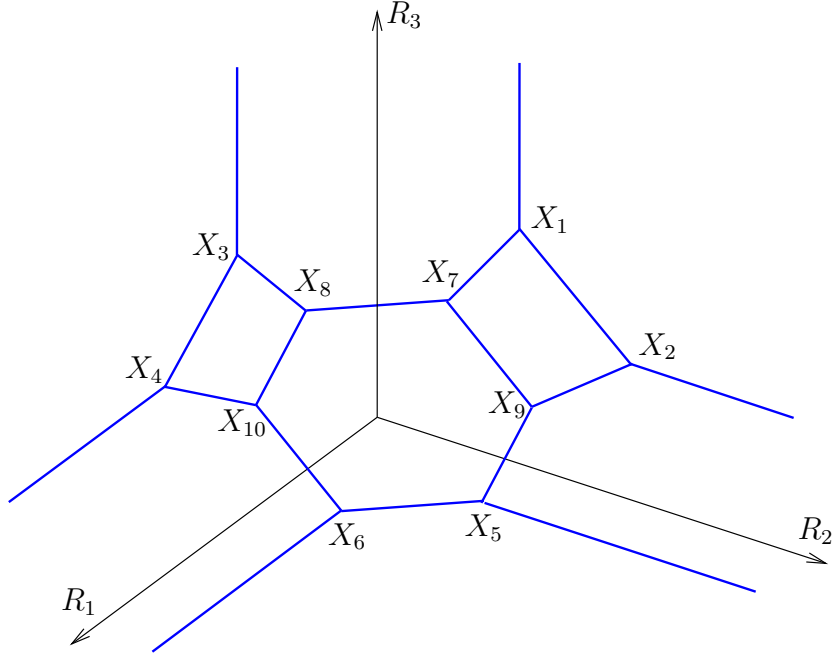}
\caption{Rate region for Regime I: $h_3\geq h_4+h_5$} \hspace{.1cm}
\label{reg1}
\end{center}
\end{figure}

\textbf{Regime II: $h_4\leq h_3 \leq h_4+h_5$}

In this regime, \eqref{r123-2} is dominated by \eqref{r123-1}. Therefore,  we only have to consider ten hyperplanes. The rates and
encoding scheme for the corner points
$Y_1=\langle$\ref{r1},\ref{r12},\ref{r1123}$\rangle$, 
$Y_2=\langle$\ref{r1},\ref{r13},\ref{r1123}$\rangle$, 
$Y_3=\langle$\ref{r2},\ref{r12},\ref{r1223}$\rangle $, 
$Y_4=\langle$\ref{r2},\ref{r23},\ref{r1223}$\rangle$, 
$Y_7=\langle$\ref{r12},\ref{r1123},\ref{r123-1}$\rangle$, 
$Y_8=\langle$\ref{r12},\ref{r1223},\ref{r123-1}$\rangle$, 
$Y_9=\langle$\ref{r13},\ref{r1123},\ref{r123-1}$\rangle$ and
$Y_{10}=\langle$\ref{r23},\ref{r1223},\ref{r123-1}$\rangle$ are exactly
the same as that of $X_1$, $X_2$, $X_3$, $X_4$, $X_7$, $X_8$, $X_9$, and
$X_{10}$, respectively. For the remaining corner points, we next provide the encoding schemes.

\begin{itemize}
\item $Y_5=\langle$\ref{r3}, \ref{r13}, \ref{r1233}$\rangle:(H_1+h_4+h_5, H_2+h_4+h_5+h_6+h_7, H_3)$

Partition $\V_5$ into $\V_{5,1}$ and $\V_{5,2}$ of lengths
$\ell_3-\ell_4$ and $\ell_4+\ell_5-\ell_3$, respectively. Also
partition $\V_3$ into $\V_{3,1}$ and $\V_{3,2}$, of sizes
$\ell_3-\ell_4$ and $\ell_4$, respectively.
\begin{eqnarray*}
\Gamma_1:\V_1,\V_4,\V_5, \qquad \Gamma_2:\V_1,\V_2,\V_{3,2}\oplus\V_4,
\V_{3,1}\oplus \V_{5,1}, \V_{5,2},\V_6,\V_7, \qquad \Gamma_3:\V_1,\V_2,\V_3.
\end{eqnarray*}

\item $Y_6=\langle$\ref{r3}, \ref{r23}, \ref{r1233}$\rangle:(H_1+h_4+h_5+h_7, H_2+h_4+h_5+h_6, H_3 )$

Partition $\V_5$ into $\V_{5,1}$ and $\V_{5,2}$ of lengths
$\ell_3-\ell_4$ and $\ell_4+\ell_5-\ell_3$, respectively. Also partition
$\V_3$ into $\V_{3,1}$ and $\V_{3,2}$, of sizes $\ell_3-\ell_4$ and
$\ell_4$, respectively.
\begin{eqnarray*}
\Gamma_1:\V_1,\V_4,\V_5,\V_7, \qquad \Gamma_2:\V_1,\V_2,\V_{3,2}\oplus\V_4,
\V_{3,1}\oplus \V_{5,1}, \V_{5,2},\V_6, \qquad \Gamma_3:\V_1,\V_2,\V_3.
\end{eqnarray*}

\item $Y_{11}=\langle$\ref{r13}, \ref{r1233}, \ref{r123-1}$\rangle:(H_1+h_3, H_3+h_6+h_7, H_2+h_4+h_5 )$

Partition $\V_5$ into $\V_{5,1}$ and $\V_{5,2}$ of lengths
$\ell_3-\ell_4$ and $\ell_4+\ell_5-\ell_3$, respectively. Also
partition $\V_3$ into $\V_{3,1}$ and $\V_{3,2}$, of sizes
$\ell_3-\ell_4$ and $\ell_4$, respectively.
\begin{eqnarray*}
\Gamma_1:\V_1,\V_4,\V_{5,1}, \qquad \Gamma_2:\V_1,\V_2,\V_{3,2}\oplus\V_4,
\V_{3,1}\oplus \V_{5,1},\V_6,\V_7, \qquad \Gamma_3:\V_1,\V_2,\V_3,\V_{5,2}.
\end{eqnarray*}

\item $Y_{12}=\langle$\ref{r23}, \ref{r1233}, \ref{r123-1}$\rangle:(H_1+h_3+h_7, H_3+h_6, H_2+h_4+h_5 )$

Partition $\V_5$ into $\V_{5,1}$ and $\V_{5,2}$ of lengths
$\ell_3-\ell_4$ and $\ell_4+\ell_5-\ell_3$, respectively. Also
partition $\V_3$ into $\V_{3,1}$ and $\V_{3,2}$, of sizes
$\ell_3-\ell_4$ and $\ell_4$, respectively.
\begin{eqnarray*}
\Gamma_1:\V_1,\V_4,\V_{5,1},\V_7, \qquad \Gamma_2:\V_1,\V_2,\V_{3,2}\oplus\V_4,
\V_{3,1}\oplus \V_{5,1},\V_6, \qquad \Gamma_3:\V_1,\V_2,\V_3,\V_{5,2}.
\end{eqnarray*}
Fig.~\ref{reg2} shows the rate region for this regime.
\begin{figure}[ht]
\begin{center}
   \psfrag{x1}[Bc][Bc]{$Y_1$}
   \psfrag{x2}[Bc][bc]{$Y_2$}
   \psfrag{x3}[Bc][bc]{$Y_3$}
   \psfrag{x4}[Bc][bc]{$Y_4$}
   \psfrag{x5}[Bc][bc]{$Y_5$}
   \psfrag{x6}[Bc][bc]{$Y_6$}
   \psfrag{x7}[Bc][bc]{$Y_7$}
   \psfrag{x8}[Bc][bc]{$Y_8$}
   \psfrag{x9}[Bc][bc]{$Y_9$}
   \psfrag{x10}[Bc][bc]{$Y_{10}$}
   \psfrag{x11}[Bc][bc]{$Y_{11}$}
   \psfrag{x12}[Bc][bc]{$Y_{12}$}
   \psfrag{r1}[Bc][bc]{$R_1$}
   \psfrag{r2}[Bc][bc]{$R_2$}
   \psfrag{r3}[Bc][bc]{$R_3$}
   \includegraphics[width=11cm]{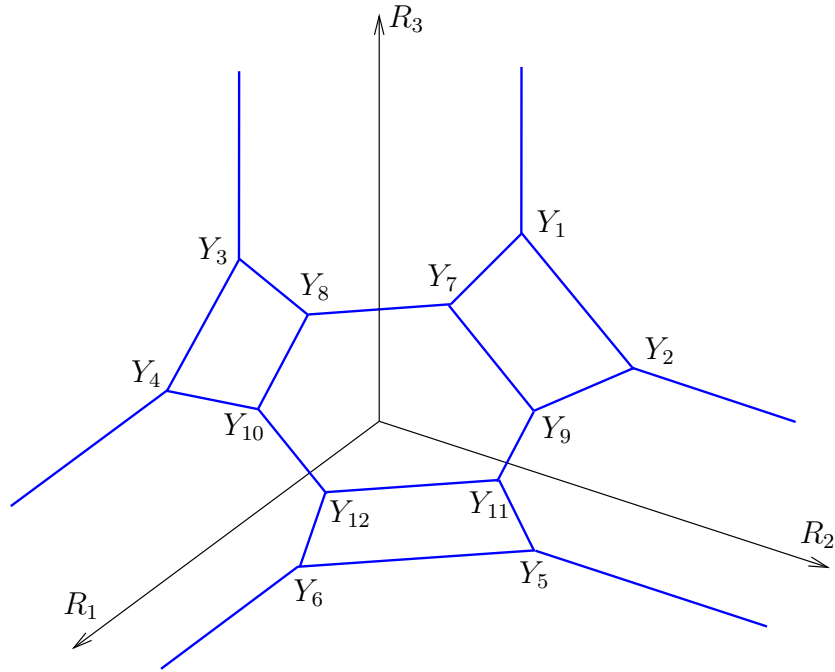}
\caption{Rate region for Regime II of ordering level $\L_1$: $h_4\leq h_3\leq h_4+h_5$} \hspace{.1cm}
\label{reg2}
\end{center}
\end{figure}
\end{itemize}

\textbf{Regime III: $h_3\leq h_4$}\\
It is clear that in this regime  \eqref{r123-1} is dominated by 
\eqref{r123-2}, and thus \eqref{r123-1} does not affect the rate region. The remaining ten inequalities 
characterize the region. The rates and coding
schemes for the points
$Z_1=\langle$\ref{r1},\ref{r12},\ref{r1123}$\rangle$, 
$Z_2=\langle$\ref{r1},\ref{r13},\ref{r1123}$\rangle$, 
$Z_3=\langle$\ref{r2},\ref{r12},\ref{r1223}$\rangle $, and
$Z_4=\langle$\ref{r2},\ref{r23},\ref{r1223}$\rangle$ are exactly the
same as that of $X_1$, $X_2$, $X_3$, and $X_4$, respectively.  The
rate tuples and the corresponding descriptions for the other corner points
are as follows.
\begin{itemize}
\item $Z_5=\langle$\ref{r3}, \ref{r13}, \ref{r1233}$\rangle:(H_1+h_4+h_5,H_2+h_4+h_5+h_6+h_7, H_3)$

Partition $\V_4$ into $\V_{4,1}$ and $\V_{4,2}$ of lengths $\ell_3$
and $\ell_4-\ell_3$, respectively.
\begin{eqnarray*}
\Gamma_1:\V_1,\V_4,\V_5, \qquad
\Gamma_2:\V_1,\V_2,\V_3\oplus\V_{4,1},\V_{4,2},\V_5,\V_6,\V_7, \qquad
\Gamma_3:\V_1,\V_2,\V_3. \\[-4.5mm]
\end{eqnarray*}
\item $Z_6=\langle$\ref{r3}, \ref{r23}, \ref{r1233}$\rangle:(H_1+h_4+h_5+h_7, H_2+h_4+h_5+h_6,  H_3 )$

Partition $\V_4$ into $\V_{4,1}$ and $\V_{4,2}$ of lengths $\ell_3$ and $\ell_4-\ell_3$, respectively.
\begin{eqnarray*}
\Gamma_1:\V_1,\V_4,\V_5,\V_7, \qquad \Gamma_2:\V_1,\V_2,\V_3\oplus\V_{4,1},
\V_{4,2},\V_5,\V_6, \qquad \Gamma_3:\V_1,\V_2,\V_3.\\[-4.5mm]
\end{eqnarray*}
\item $Z_7=\langle$\ref{r12}, \ref{r1123}, \ref{r1223},
\ref{r123-2}$\rangle:\left(H_1+\frac{h_3+h_4}{2}, H_2+\frac{h_3+h_4}{2},
H_2+\frac{h_3+h_4}{2}+h_5+h_6+h_7\right)$

Partition $\V_4$ into $\V_{4,1}$, $\V_{4,2}$, and $\V_{4,3}$ of
lengths $\ell_3$, $\frac{1}{2}(\ell_4-\ell_3)$ and
$\frac{1}{2}(\ell_4-\ell_3)$, respectively.
\begin{eqnarray*}
\Gamma_1:\V_1,\V_{4,1},\V_{4,2}  \quad \Gamma_2:\V_1,\V_2,\V_3\oplus\V_{4,1}, \V_{4,2}\oplus\V_{4,3},
\quad
\Gamma_3:\V_1,\V_2,\V_3, \V_{4,3}, \V_5,\V_6,\V_7.\\[-4.5mm]
\end{eqnarray*}
\item $Z_8=\langle$\ref{r13}, \ref{r1123}, \ref{r1233},
\ref{r123-2}$\rangle:\left(H_1+\frac{h_3+h_4}{2},H_2+\frac{h_3+h_4}{2}+h_6+h_7,H_2+\frac{h_3+h_4}{2}+h_5\right)$

Partition $\V_4$ into $\V_{4,1}$, $\V_{4,2}$, and $\V_{4,3}$ of
lengths $\ell_3$, $\frac{1}{2}(\ell_4-\ell_3)$ and
$\frac{1}{2}(\ell_4-\ell_3)$, respectively.
\begin{eqnarray*}
\Gamma_1:\V_1,\V_{4,1},\V_{4,2}  \quad \Gamma_2:\V_1,\V_2,\V_3\oplus\V_{4,1}, \V_{4,2}\oplus\V_{4,3},\V_6,\V_7,
\quad
\Gamma_3:\V_1,\V_2,\V_3, \V_{4,3}, \V_5. \\[-4.5mm]
\end{eqnarray*}
\item $Z_9=\langle$\ref{r23}, \ref{r1223},
\ref{r123-2}$\rangle:\left(H_1+\frac{h_3+h_4}{2}+h_7,
H_2+\frac{h_3+h_4}{2}, H_2+\frac{h_3+h_4}{2}+h_5+h_6\right)$

Partition $\V_4$ into $\V_{4,1}$, $\V_{4,2}$, and $\V_{4,3}$ of
lengths $\ell_3$, $\frac{1}{2}(\ell_4-\ell_3)$ and
$\frac{1}{2}(\ell_4-\ell_3)$, respectively.
\begin{eqnarray*}
\Gamma_1:\V_1,\V_{4,1},\V_{4,2},\V_7  \quad \Gamma_2:\V_1,\V_2,\V_3\oplus\V_{4,1},\V_{4,2}\oplus\V_{4,3},
\quad
\Gamma_3:\V_1,\V_2,\V_3, \V_{4,3}, \V_5,\V_6. \\[-4mm]
\end{eqnarray*}
\item $Z_{10}=\langle$\ref{r23}, \ref{r1223},
\ref{r123-2}$\rangle:\left(H_1+\frac{h_3+h_4}{2}+h_7,
H_2+\frac{h_3+h_4}{2}+h_6, H_2+\frac{h_3+h_4}{2}+h_5\right)$

Partition $\V_4$ into $\V_{4,1}$, $\V_{4,2}$, and $\V_{4,3}$ of
lengths $\ell_3$, $\frac{1}{2}(\ell_4-\ell_3)$ and
$\frac{1}{2}(\ell_4-\ell_3)$, respectively.
\begin{eqnarray*}
\Gamma_1:\V_1,\V_{4,1},\V_{4,2},\V_7  \quad \Gamma_2:\V_1,\V_2,\V_3\oplus\V_{4,1},\V_{4,2}\oplus\V_{4,3},\V_6,
\quad
\Gamma_3:\V_1,\V_2,\V_3, \V_{4,3}, \V_5.
\end{eqnarray*}

This region and its corner points are shown in Fig.~\ref{reg3}.
\begin{figure}[ht]
\begin{center}
   \psfrag{x1}[Bc][Bc]{$Z_1$}
   \psfrag{x2}[Bc][bc]{$Z_2$}
   \psfrag{x3}[Bc][bc]{$Z_3$}
   \psfrag{x4}[Bc][bc]{$Z_4$}
   \psfrag{x5}[Bc][bc]{$Z_5$}
   \psfrag{x6}[Bc][bc]{$Z_6$}
   \psfrag{x7}[Bc][bc]{$Z_7$}
   \psfrag{x8}[Bc][bc]{$Z_8$}
   \psfrag{x9}[Bc][bc]{$Z_9$}
   \psfrag{x10}[Bc][bc]{$Z_{10}$}
   \psfrag{r1}[Bc][bc]{$R_1$}
   \psfrag{r2}[Bc][bc]{$R_2$}
   \psfrag{r3}[Bc][bc]{$R_3$}
   \includegraphics[width=11cm]{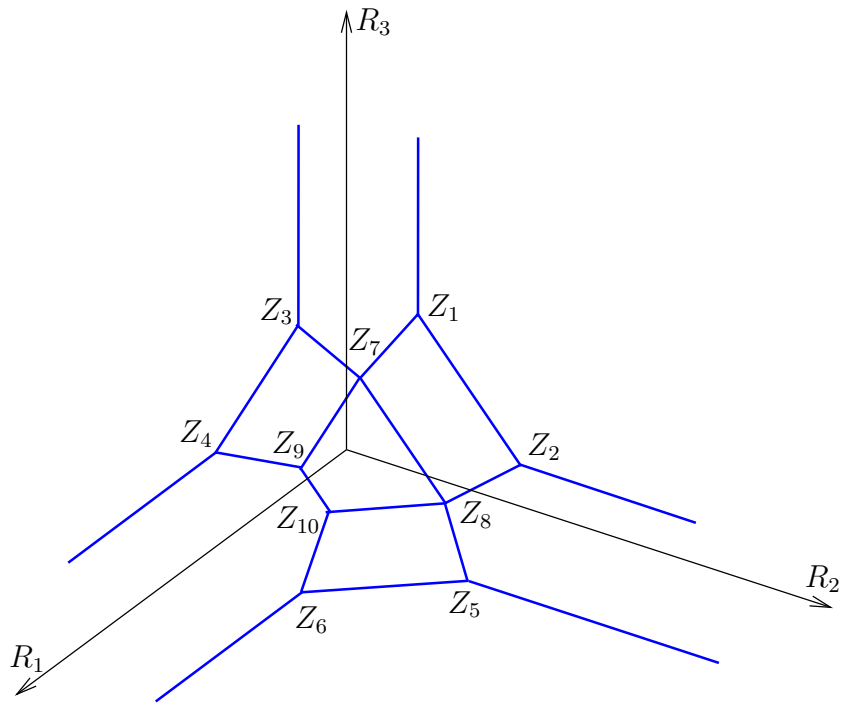}
\caption{Rate region for Regime III of the ordering level $\L_1$: $h_3\leq h_4$} \hspace{.1cm}
\label{reg3}
\end{center}
\end{figure}
\end{itemize}

The coding schemes proposed for these three cases give us the achievability proof of the theorem for the specific ordering  $\L_1$. As stated before, the coding scheme for other possible orderings listed in Table~\ref{tbl:ordering} are similar to that of the ordering $\L_1$. There are three main ingredients used in all of them; (1) converting the source sequences into bitstreams,  (2) partitioning the bit streams into sequences of proper length,  and (3) (if required) applying linear coding (binary \texttt{xor}) on them.  
This completes the proof of Theorem~\ref{thm:mld}. 

\section{Asymmetric Multiple Description Coding}

In this section we prove Theorems~\ref{thm:md-outer} and \ref{thm:md-inner}, which together give an approximate characterization for the admissible rate region of the A-MD problem.

\subsection{An Outer Bound for the Rate Region of A-MD: Proof of Theorem~\ref{thm:md-outer}}
 
In order to prove this theorem, we first show a parametric outer-bound for the A-MD rate  region. Then we specialize the parameters to obtain the bound claimed in the theorem. 

We first need to define a set of auxiliary random variables in order to state and prove the parametric bound, which are some noisy versions of the source. The strategy of expanding the probability space by a single auxiliary variable was used to characterize the two descriptions Gaussian MD region \cite{ozarow}, and later in \cite{TianMohDig_j} extended to include multiple auxiliary random variables with certain built-in Markov structure. We shall continue to use this extended strategy as used in \cite{TianMohDig_j}.

Let $N_i \sim \N(0,\sigma_i^2),\ i=1,\dots,6$, be mutually independent zero-mean Gaussian random variables with variance $\sigma_i^2$. They are also  assumed to be independent of $X$.  A noisy version of the source, $Y_i$, is defined as 
\begin{eqnarray}
Y_i=X+Z_i, \quad i=1,\dots,6 
\label{eq:def-y}
\end{eqnarray}
where $Z_i=\sum_{j=i}^6 N_j$ for $j=1,\dots,6$. Thus $d_i\triangleq \sum_{j=i}^6 \sigma_j^2$ would be the variance of the noises $Z_i$, for $i=1,\dots,6$. We also define  $Y_7=X$ and $d_7=0$ for convenience. Note that incremental noises are added to $X$ to build $Y_i$'s, and therefore they form a Markov chain as 
\begin{eqnarray}
(\Gamma_1,\Gamma_2,\Gamma_3)\leftrightarrow X^n \leftrightarrow Y_{6}^n \leftrightarrow Y_{5}^n \leftrightarrow \cdots \leftrightarrow Y_{1}^n.
\label{eq:Markov}
\end{eqnarray}

The following theorem provides a parametric  outer-bound for the rate region of the A-MLD problem, depending on $d_i$ variables, which are the noise variances defined above. Such bound holds for any choice of $d_1\geq d_2\geq \dots\geq d_6>0$, and can be further optimized to obtain a good non-parametric outer-bound for the rate region. However, we simply derive the bound in Theorem~\ref{thm:md-outer} by setting the values of $d_i$'s.

\begin{theorem}
For a given distortion vector $\D=(D_{\Gamma_1},\dots,D_{\Gamma_1\Gamma_2\Gamma_3})$ and a set of variables $d_1\geq d_2\geq \dots\geq d_6 > d_7=0$, denote by $\uR\MD^p(\D,\mathbf{d})$ the set of all rate triples $(R_1,R_2,R_3)$ satisfying
{\makeatletter
\newcounter{tempcntg}
\setcounter{tempcntg}{\the\c@equation}
\def\theequation{$\mathcal{PO}-$\@arabic\c@equation}\makeatother
\setcounter{equation}{0}
\begin{align}
R_j &\geq  \frac{1}{2} \log \frac{1}{D_{\Gamma_j}} \qquad j=1,2,3 \label{upR1}\\
R_i+R_j & \geq  \frac{1}{2} \log \frac{1+d_{\L(\Gamma_i) } }{D_{\Gamma_i}+d_{\L(\Gamma_i)}} \frac{1+d_{\L(\Gamma_j)}}{D_{\Gamma_j}+d_{\L(\Gamma_j)}} \frac{(D_{\Gamma_i \Gamma_j} +d_{\max\{\L(\Gamma_i),\L(\Gamma_j) \}})} {(1+d_{\max\{\L(\Gamma_i),\L(\Gamma_j) \}})D_{\Gamma_i \Gamma_j}} \qquad i\neq j \label{upR12}\\
2R_i+R_j +R_k & \geq \frac{1}{2} \log \left(\frac{1+d_{\L(\Gamma_i)}}{D_{\Gamma_i} + d_{\L(\Gamma_i)}}\right)^2  \frac{1+d_{\L(\Gamma_j)}}{D_{\Gamma_j} + d_{\L(\Gamma_j)}}   \frac{1+d_{\L(\Gamma_k)}}{D_{\Gamma_k} + d_{\L(\Gamma_k)}} \nonumber\\
&\phantom{\geq} +\frac{1}{2} \log \frac{(1+d_{\L(\Gamma_i, \Gamma_j) } ) (D_{\Gamma_i \Gamma_j}  + d_{\max\{\L(\Gamma_i), \L(\Gamma_j)\}}) } {(1+ d_{\max\{\L(\Gamma_i),\L(\Gamma_j)\} } )(D_{\Gamma_i \Gamma_j}+d_{\L(\Gamma_i, \Gamma_j)})}\nonumber\\
&\phantom{\geq} +\frac{1}{2} \log\frac{(1+d_{\L(\Gamma_i, \Gamma_k) } ) (D_{\Gamma_i \Gamma_k}  + d_{\max\{\L(\Gamma_i), \L(\Gamma_k)\}}) } {(1+ d_{\max\{\L(\Gamma_i),\L(\Gamma_k)\} } )(D_{\Gamma_i \Gamma_k}+d_{\L(\Gamma_i, \Gamma_k)})}
\nonumber\\
&\phantom{\geq} +\frac{1}{2} \log \frac{ (D_{\Gamma_i \Gamma_j \Gamma_k}  + d_{\max\{\L(\Gamma_i, \Gamma_j), \L(\Gamma_i, \Gamma_k)\}}) } {(1+ d_{\max\{\L(\Gamma_i, \Gamma_j),\L(\Gamma_i, \Gamma_k)\} } )D_{\Gamma_i \Gamma_j \Gamma_k}}
\qquad i\neq j\neq k\label{upR1123} \\
R_1+R_2 +R_3 &\geq   \frac{1}{2} \log \frac{1+d_{\L(\Gamma_1) } }{D_{\Gamma_1}+d_{\L(\Gamma_1)}} \frac{1+d_{\L(\Gamma_2)}}{D_{\Gamma_2}+d_{\L(\Gamma_2)}}  
\frac{1+d_{\min\{\L(\Gamma_1,\Gamma_2), \L(\Gamma_3)\} }} {D_{\Gamma_3}+d_{\min\{\L(\Gamma_1,\Gamma_2), \L(\Gamma_3)\}}}\nonumber\\
&\phantom{\geq} +\frac{1}{2} 
\log\frac{(1+d_{\min\{\L(\Gamma_1,\Gamma_2), \L(\Gamma_3)\}} )(D_{\Gamma_1 \Gamma_2} + d_{\L(\Gamma_2)}  ) }{ (1+d_{\L(\Gamma_2)}) (D_{\Gamma_1 \Gamma_2} + d_{\min\{\L(\Gamma_1,\Gamma_2), \L(\Gamma_3)\}} ) }\nonumber\\
&\phantom{\geq} +\frac{1}{2} 
\log\frac{D_{\Gamma_1 \Gamma_2 \Gamma_3}+d_{  \min\{ \L(\Gamma_1, \Gamma_2), \L(\Gamma_3)\} } }{(1+d_{  \min\{ \L(\Gamma_1, \Gamma_2), \L(\Gamma_3)\} }) D_{\Gamma_1 \Gamma_2 \Gamma_3}}\label{upR123-2}\\
R_1+R_2 +R_3 &\geq 
\frac{1}{2} \log \frac{1+d_{\L(\Gamma_1)}}{D_{\Gamma_1} + d_{\L(\Gamma_1)}}  \frac{1+d_{\L(\Gamma_2)}}{D_{\Gamma_2} + d_{\L(\Gamma_2)}}   \frac{1+d_{\L(\Gamma_3)}}{D_{\Gamma_3} + d_{\L(\Gamma_3)}} \nonumber\\
&\phantom{\geq} +\frac{1}{4}
\log\frac{(1+d_{\alpha} )(D_{\Gamma_1 \Gamma_2} + d_{\L(\Gamma_2)}  ) }{ (1+d_{\L(\Gamma_2)}) (D_{\Gamma_1 \Gamma_2} + d_{\alpha} ) }
+\frac{1}{4}
\log\frac{(1+d_{\alpha} )(D_{\Gamma_1 \Gamma_3} + d_{\L(\Gamma_3)}  ) }{ (1+d_{\L(\Gamma_3)}) (D_{\Gamma_1 \Gamma_3} + d_{\alpha} ) }
\nonumber\\
&\phantom{\geq} 
+\frac{1}{4}
\log\frac{(1+d_{\alpha} )(D_{\Gamma_2 \Gamma_3} + d_{\L(\Gamma_3)}  ) }{ (1+d_{\L(\Gamma_3)}) (D_{\Gamma_2 \Gamma_3} + d_{\alpha} ) }
+\frac{1}{2} 
\log\frac{D_{\Gamma_1 \Gamma_2 \Gamma_3}+d_{  \L(\Gamma_3) } }{(1+d_{   \L(\Gamma_3) }) D_{\Gamma_1 \Gamma_2 \Gamma_3}},\label{upR123-1}
\end{align}
\makeatletter
\setcounter{equation}{\the\c@tempcntg}\makeatother
}
where 
\begin{align*}
\alpha=\left\{\begin{array}{ll}
\L(\Gamma_3) & \textrm{if } \L(\Gamma_3)>\L(\Gamma_1,\Gamma_2),\\
\min\{\L(\Gamma_1,\Gamma_2),\L(\Gamma_1,\Gamma_3),\L(\Gamma_2,\Gamma_3)\}  & \textrm{if } \L(\Gamma_3)<\L(\Gamma_1,\Gamma_2).
\end{array}\right.
\end{align*}
Then any admissible rate triple belongs to $\uR\MD^p(\D,\mathbf{d})$, \emph{i.e.,} $\R\MD(\D)\subseteq\uR\MD^p(\D,\mathbf{d})$, for all choices of $d_1\geq d_2\geq \dots\geq d_6 \geq d_7=0$. 
\label{thm:md-outer-parametric}
\end{theorem}

The following two lemmas are extracted from \cite{TianMohDig08}, whose proofs can be found in Appendix~A for completeness. They are useful to bound the mutual information between the noisy versions of the source and the descriptions. 

\begin{lemma}\label{lemma1}
For any set of descriptions $\S\subseteq \{\Gamma_1,\Gamma_2,\Gamma_3\}$, and noisy version of the source $Y_i$, $i=1,2,\dots,7$, we have
\begin{eqnarray}
I( \S; Y_i^n ) \geq \frac{n}{2} \log \frac{1+d_i}{D_{\S} + d_i}.
\end{eqnarray}
\end{lemma}

\begin{lemma}\label{lemma2}
For any subset of the descriptions $\S$, and two noisy versions of the source $Y_i$ and $Y_j$ with $i<j$, we have
\begin{eqnarray}
I(\S ; Y_j^n) - I(\S ; Y_i^n) \geq \frac{n}{2} \log \frac{(1+d_j)(D_{\S}+d_i)}{(1+d_i)(D_{\S}+d_j)}.
\end{eqnarray}
\end{lemma}

We will use these results in several points in the proof of Theorem~\ref{thm:md-outer-parametric}, which are indicated by $(\dagger)$. Now, we are ready to prove the parametric outer-bound.

\begin{proof}[Proof of Theorem~\ref{thm:md-outer-parametric}] 
 
The single description levels inequalities are just straight forward result of Lemma~\ref{lemma1}. We have
\begin{eqnarray}
nR_i \geq H(\Gamma_i) = H(\Gamma_i)- H(\Gamma_i | X^n) = I( \Gamma_i ; X^n) \stackrel{(\dagger)}{\geq} \frac{n}{2} \log \frac{1}{D_{\Gamma_i}}
\end{eqnarray}
where we used Lemma~\ref{lemma1} for $Y_7=X$ and the fact $d_7=0$ in the last inequality. This proves (\ref{upR1}).

The bound for the two description rates in (\ref{upR12}) follows from 
\begin{align}
n(R_i+R_j&+2\e) \geq  H(\Gamma_i) + H(\Gamma_j)\nonumber\\
&\stackrel{(a)}{\geq}  H(\Gamma_i) + H(\Gamma_j) -H(\Gamma_i,\Gamma_j | X^n) \nonumber\\
&\phantom{=}- \left[H(\Gamma_i | Y_{\max\{\L(\Gamma_i),\L(\Gamma_j) \}}^n) + H(\Gamma_j | Y_{\max\{\L(\Gamma_i),\L(\Gamma_j) \}}^n) -H(\Gamma_i,\Gamma_j | Y_{\max\{\L(\Gamma_i),\L(\Gamma_j) \}}^n)\right] \nonumber\\
&= I(\Gamma_i; Y_{\max\{\L(\Gamma_i),\L(\Gamma_j) \}}^n) + I(\Gamma_j; Y_{\max\{\L(\Gamma_i),\L(\Gamma_j) \}}^n) \nonumber\\
&\phantom{=} + [I(\Gamma_i\Gamma_j; X^n) -I(\Gamma_i\Gamma_j; Y_{\max\{\L(\Gamma_i),\L(\Gamma_j) \}}^n)] \nonumber\\
&\stackrel{(b)}{\geq}  I(\Gamma_i; Y_{\L(\Gamma_i)}^n) + I(\Gamma_j; Y_{\L(\Gamma_j) }^n) + [I(\Gamma_i\Gamma_j; X^n) -I(\Gamma_i\Gamma_j; Y_{\max\{\L(\Gamma_i),\L(\Gamma_j) \}}^n)]\nonumber\\
&\stackrel{(\dagger)}{\geq}  \frac{n}{2} \log \frac{1+d_{\L(\Gamma_i) } }{D_{\Gamma_i}+d_{\L(\Gamma_i)}} \frac{1+d_{\L(\Gamma_j)}}{D_{\Gamma_j}+d_{\L(\Gamma_j)}} \frac{(1+0)(D_{\Gamma_i \Gamma_j} +d_{\max\{\L(\Gamma_i),\L(\Gamma_j) \}})} {(1+d_{\max\{\L(\Gamma_i),\L(\Gamma_j) \}})(D_{\Gamma_i \Gamma_j}+0)}
\label{eqn:R1R2}
\end{align}
where the subtracted terms in $(a)$ are positive due to the fact that $\Gamma_1$ and $\Gamma_2$ are functions of $X^n$ and non-negativity of mutual information, $(b)$ is by the data processing inequality and the Markov chain in \eqref{eq:Markov}. Finally, we have used Lemma~\ref{lemma1} and Lemma~\ref{lemma2} in $(\dagger)$.

The inequality (\ref{upR1123}) can be proved through the following chain of inequalities. 
\begin{align}
n(2R_i&+R_j+ R_k+4\e) 
\geq  
2H(\Gamma_i)+H(\Gamma_j)+H(\Gamma_k)\nonumber\\
&\stackrel{(a)}{\geq} 
2H(\Gamma_i)+H(\Gamma_j)+H(\Gamma_k)- H(\Gamma_i\Gamma_j \Gamma_k | X^n) \nonumber\\
&\phantom{\geq}-\left[H(\Gamma_i | Y_{\max\{\L(\Gamma_i),\L(\Gamma_j) \}}^n) +H(\Gamma_j | Y_{\max\{\L(\Gamma_i),\L(\Gamma_j) \}}^n)  -H(\Gamma_i \Gamma_j | Y_{\max\{\L(\Gamma_i),\L(\Gamma_j) \}}^n) \right]\nonumber\\ 
&\phantom{\geq}-\left[H(\Gamma_i | Y_{\max\{\L(\Gamma_i),\L(\Gamma_k) \}}^n) +H(\Gamma_k | Y_{\max\{\L(\Gamma_i),\L(\Gamma_k) \}}^n)  -H(\Gamma_i \Gamma_k | Y_{\max\{\L(\Gamma_i),\L(\Gamma_k) \}}^n) \right] \nonumber\\
&\phantom{\geq}-\Big[ H(\Gamma_i \Gamma_j | Y_{\max\{\L(\Gamma_i, \Gamma_j),\L(\Gamma_i, \Gamma_k) \}}^n) + H(\Gamma_i \Gamma_k | Y_{\max\{\L(\Gamma_i, \Gamma_j),\L(\Gamma_i, \Gamma_k) \} }^n)  \nonumber\\
&\phantom{=-\Big[ } -H(\Gamma_i \Gamma_j \Gamma_k | Y_{\max\{\L(\Gamma_i, \Gamma_j),\L(\Gamma_i, \Gamma_k) \}}^n)\Big]\nonumber\\
&= I(\Gamma_i ; Y_{\max\{\L(\Gamma_i),\L(\Gamma_j) \}}^n) + I(\Gamma_j ; Y_{\max\{\L(\Gamma_i),\L(\Gamma_j) \}}^n) \nonumber\\
&\phantom{=}+  I(\Gamma_i ; Y_{\max\{\L(\Gamma_i),\L(\Gamma_k) \}}^n)
+  I(\Gamma_k ; Y_{\max\{\L(\Gamma_i),\L(\Gamma_k) \}}^n) \nonumber\\
&\phantom{=}+  [ I( \Gamma_i \Gamma_j ; Y_{\max\{\L(\Gamma_i, \Gamma_j),\L(\Gamma_i, \Gamma_k) \}}^n) -I(\Gamma_i \Gamma_j ; Y_{\max\{\L(\Gamma_i),\L(\Gamma_j) \}}^n)]   \nonumber\\
&\phantom{=}+[I(\Gamma_i \Gamma_k ; Y_{\max\{\L(\Gamma_i, \Gamma_j),\L(\Gamma_i, \Gamma_k) \}}^n)- I(\Gamma_i \Gamma_k ; Y_{\max\{\L(\Gamma_i),\L(\Gamma_k) \}}^n)] \nonumber\\
&\phantom{=} +[I(\Gamma_i \Gamma_j \Gamma_k ; X^n) - I(\Gamma_i \Gamma_j \Gamma_k ; Y_{\max\{\L(\Gamma_i, \Gamma_j),\L(\Gamma_i, \Gamma_k) \}}^n)] \label{eq:pr:upR1123-1}
\end{align}
where in $(a)$ we have used the fact that all the brackets are non-negative.
Now, we will bound each term in \eqref{eq:pr:upR1123-1} individually. The single description terms can be bounded as
\begin{align}
I(\Gamma_i ; Y_{\max\{\L(\Gamma_i),\L(\Gamma_k) \}}^n) &\geq  I(\Gamma_i ; Y_{\L(\Gamma_i)}^n) 
\stackrel{(\dagger)}{\geq} \frac{n}{2} \log \frac{1+d_{\L(\Gamma_i)}}{D_{\Gamma_i} + d_{\L(\Gamma_i)}},
  \label{eq:pr:upR1123-2}
\end{align}
and similarly for $j$ and $k$. Also we can bound the differential terms as
\begin{align}
 I( \Gamma_i \Gamma_j ; Y_{\max\{\L(\Gamma_i, \Gamma_j),\L(\Gamma_i, \Gamma_k) \}}^n) &-I(\Gamma_i \Gamma_j ; Y_{\max\{\L(\Gamma_i),\L(\Gamma_j) \}}^n) \nonumber\\
 &\stackrel{(b)}{\geq}  
   I( \Gamma_i \Gamma_j ; Y_{\L(\Gamma_i \Gamma_j)}^n) -I(\Gamma_i \Gamma_j ; Y_{\max\{\L(\Gamma_i),\L(\Gamma_j) \}}^n)\nonumber\\
&\stackrel{(\dagger)}{\geq}  \frac{n}{2} \log \frac{(1+d_{\L(\Gamma_i \Gamma_j) } ) (D_{\Gamma_i \Gamma_j}  + d_{\max\{\L(\Gamma_i), \L(\Gamma_j)\}}) } {(1+ d_{\max\{\L(\Gamma_i),\L(\Gamma_j)\} } )(D_{\Gamma_i \Gamma_j}+d_{\L(\Gamma_i \Gamma_j)})}
\label{eq:pr:upR1123-3}
\end{align}
where $(b)$ is due to the data processing inequality implied by the Markov chain 
\begin{eqnarray*}
(\Gamma_i \Gamma_j) \leftrightarrow Y_{\max\{\L(\Gamma_i, \Gamma_j),\L(\Gamma_i, \Gamma_k) \}}^n \leftrightarrow Y_{\L(\Gamma_i \Gamma_j)}^n
\end{eqnarray*}
implied by \eqref{eq:Markov}. We also have
\begin{align}
  I( \Gamma_i \Gamma_j \Gamma_k ; X^n) &-I(\Gamma_i \Gamma_j \Gamma_k; Y_{\max\{\L(\Gamma_i, \Gamma_j),\L(\Gamma_i, \Gamma_k) \}}^n) \nonumber\\
&\stackrel{(\dagger)}{\geq}  \frac{n}{2} \log \frac{(1+0) (D_{\Gamma_i \Gamma_j \Gamma_k}  + d_{\max\{\L(\Gamma_i, \Gamma_j), \L(\Gamma_i, \Gamma_k)\}}) } {(1+ d_{\max\{\L(\Gamma_i, \Gamma_j),\L(\Gamma_i, \Gamma_k)\} } )(D_{\Gamma_i \Gamma_j \Gamma_k}+0)}.
\label{eq:pr:upR1123-4}
\end{align}
By replacing \eqref{eq:pr:upR1123-2}--\eqref{eq:pr:upR1123-4} in \eqref{eq:pr:upR1123-1} we get the desired inequality.

In order to derive the sum-rate bound in \eqref{upR123-2}, we can write 
\begin{align}
n(R_1+R_2+&R_3+3\e) \geq  H(\Gamma_1)+H(\Gamma_2)+H(\Gamma_3) \nonumber\\
&\geq  H(\Gamma_1)+H(\Gamma_2)+H(\Gamma_3) -H(\Gamma_1 \Gamma_2 \Gamma_3 | X^n)\nonumber\\
&\phantom{\geq}- \left[H(\Gamma_1 | Y_{\L(\Gamma_2)}^n) +H(\Gamma_2 | Y_{\L(\Gamma_2)}^n)  -H(\Gamma_1 \Gamma_2 | Y_{\L(\Gamma_2)}^n) \right] \nonumber\\
&\phantom{\geq}-\Big[ H(\Gamma_1 \Gamma_2 | Y_{\min\{\L(\Gamma_1,\Gamma_2), \L(\Gamma_3)\} }^n) + H(\Gamma_3 | Y_{\min\{\L(\Gamma_1,\Gamma_2), \L(\Gamma_3)\} }^n) \nonumber\\
&\phantom{\geq -\Big[}-H(\Gamma_1 \Gamma_2 \Gamma_3 | Y_{\min\{\L(\Gamma_1,\Gamma_2), \L(\Gamma_3)\} }^n) \Big] \nonumber\\
&\geq  I(\Gamma_1 ; Y_{\L(\Gamma_1)}^n) + I(\Gamma_2 ; Y_{\L(\Gamma_2)}^n) + I(\Gamma_3 ; Y_{\min\{\L(\Gamma_1, \Gamma_2), \L(\Gamma_3)\}}^n) \nonumber\\
&\phantom{\geq}+ \left[I(\Gamma_1 \Gamma_2 ; Y_{\min\{\L(\Gamma_1,\Gamma_2), \L(\Gamma_3)\} }^n)- I(\Gamma_1 \Gamma_2 ; Y_{\L(\Gamma_2)}^n)\right]\nonumber\\
&\phantom{=} + \left[I(\Gamma_1 \Gamma_2 \Gamma_3; X^n)- I(\Gamma_1 \Gamma_2 \Gamma_3 ; Y_{\min\{\L(\Gamma_1,\Gamma_2), \L(\Gamma_3)\} }^n)\right].
\label{eq:pr:upR123-2}
\end{align}
Again, applying Lemma~\ref{lemma1} and Lemma~\ref{lemma2} we can bound each term in \eqref{eq:pr:upR123-2}, and obtain \eqref{upR123-2}. 

It remains to show the bound in \eqref{upR123-1}. Recall the proof of \eqref{ur123-1}, and consider two cases. 
If $\L(\Gamma_3)>\L(\Gamma_1, \Gamma_2)$, then 
Using the similar argument as in the proof of \eqref{upR1123}, we obtain
\begin{align}
n(R_1+R_2+R_3+3\e) 
&\geq  H(\Gamma_1)+H(\Gamma_2)+H(\Gamma_3) -H(\Gamma_1 \Gamma_2 \Gamma_3 | X^n)\nonumber\\
&\phantom{\geq} - \frac{1}{2}\left[H(\Gamma_1 | Y_{\L(\Gamma_2)}^n) +H(\Gamma_2 | Y_{\L(\Gamma_2)}^n)  -H(\Gamma_1 \Gamma_2 | Y_{\L(\Gamma_2)}^n) \right]\nonumber\\ 
&\phantom{\geq} -\frac{1}{2}\left[H(\Gamma_1 | Y_{\L(\Gamma_3)}^n) +H(\Gamma_3 | Y_{\L(\Gamma_3)}^n)  -H(\Gamma_1 \Gamma_3 | Y_{\L(\Gamma_3)}^n) \right] \nonumber\\
&\phantom{\geq}-\frac{1}{2}\left[H( \Gamma_2 | Y_{\L(\Gamma_3)}^n) +H(\Gamma_3 | Y_{\L(\Gamma_3)}^n)  -H(\Gamma_2 \Gamma_3 | Y_{\L(\Gamma_3)}^n) \right]\nonumber\\
&\phantom{\geq}  - \frac{1}{2} \left[ H(\Gamma_1 \Gamma_2 | Y_{\L(\Gamma_3)}^n) + H(\Gamma_1 \Gamma_3 | Y_{\L(\Gamma_3)}^n) + H(\Gamma_2 \Gamma_3 | Y_{\L(\Gamma_3)}^n) \right.\nonumber\\
&\phantom{\geq  - \frac{1}{2} \left[\right.} \left.  - 2 H(\Gamma_1 \Gamma_2 \Gamma_3 | Y_{\L(\Gamma_3)}^n) \right]\nonumber\\
&\geq  I(\Gamma_1 ; Y_{\L(\Gamma_1)}^n) + I(\Gamma_2 ; Y_{\L(\Gamma_2)}^n) + I(\Gamma_3 ; Y_{\L(\Gamma_3)}^n)\nonumber\\
&\phantom{\geq} + \frac{1}{2}  [I(\Gamma_1 \Gamma_2 ; Y_{\L(\Gamma_3)}^n) - I(\Gamma_1 \Gamma_2 ; Y_{\L(\Gamma_2)}^n)]  \nonumber\\
&\phantom{\geq} +[I (\Gamma_1 \Gamma_2 \Gamma_3 ; X^n) - I(\Gamma_1 \Gamma_2 \Gamma_3 ; Y_{\L(\Gamma_3)}^n)], 
\label{pr:out10-1}
\end{align}
which gives us the desired inequality by using Lemma~\ref{lemma1} and Lemma~\ref{lemma2} to bound each individual term. Similarly, for the case where $\L(\Gamma_3)<\L(\Gamma_1, \Gamma_2)$ we can write

\begin{align}
n(R_1+R_2+R_3+3\e) 
&\geq  H(\Gamma_1)+H(\Gamma_2)+H(\Gamma_3) -H(\Gamma_1 \Gamma_2 \Gamma_3 | X^n)\nonumber\\
&\phantom{\geq} - \frac{1}{2}\left[H(\Gamma_1 | Y_{\L(\Gamma_2)}^n) +H(\Gamma_2 | Y_{\L(\Gamma_2)}^n)  -H(\Gamma_1 \Gamma_2 | Y_{\L(\Gamma_2)}^n) \right]\nonumber\\ 
&\phantom{\geq} -\frac{1}{2}\left[H(\Gamma_1 | Y_{\L(\Gamma_3)}^n) +H(\Gamma_3 | Y_{\L(\Gamma_3)}^n)  -H(\Gamma_1 \Gamma_3 | Y_{\L(\Gamma_3)}^n) \right] \nonumber\\
&\phantom{\geq}-\frac{1}{2}\left[H( \Gamma_2 | Y_{\L(\Gamma_3)}^n) +H(\Gamma_3 | Y_{\L(\Gamma_3)}^n)  -H(\Gamma_2 \Gamma_3 | Y_{\L(\Gamma_3)}^n) \right]\nonumber\\
&\phantom{\geq}  - \frac{1}{2} \left[ H(\Gamma_1 \Gamma_2 | Y_{\beta}^n) + H(\Gamma_1 \Gamma_3 | Y_{\beta}^n)  
+ H(\Gamma_2 \Gamma_3 | Y_{\beta}^n) - 2 H(\Gamma_1 \Gamma_2 \Gamma_3 | Y_{\beta}^n) \right]\nonumber\\
&\geq  I(\Gamma_1 ; Y_{\L(\Gamma_1)}^n) + I(\Gamma_2 ; Y_{\L(\Gamma_2)}^n) + I(\Gamma_3 ; Y_{\L(\Gamma_3)}^n)\nonumber\\
&\phantom{\geq} + \frac{1}{2}  [I(\Gamma_1 \Gamma_2 ; Y_{\beta}^n) - I(\Gamma_1 \Gamma_2 ; Y_{\L(\Gamma_2)}^n)] \nonumber\\
&\phantom{\geq} + 
\frac{1}{2}  [I(\Gamma_1 \Gamma_3 ; Y_{\beta}^n) - I(\Gamma_1 \Gamma_3 ; Y_{\L(\Gamma_3)}^n)] \nonumber\\
&\phantom{\geq} + 
\frac{1}{2}  [I(\Gamma_2 \Gamma_3 ; Y_{\beta}^n) - I(\Gamma_2 \Gamma_3 ; Y_{\L(\Gamma_3)}^n)]  \nonumber\\
&\phantom{\geq} +
[I (\Gamma_1 \Gamma_2 \Gamma_3 ; X^n) - I(\Gamma_1 \Gamma_2 \Gamma_3 ; Y_{\beta}^n)].
\label{pr:out10-3}
\end{align}
where $\beta=\min\{ \L(\Gamma_1,\Gamma_2),\L(\Gamma_1,\Gamma_3),\L(\Gamma_2,\Gamma_3)\}$. Now, we can use the above-mentioned lemmas again to bound each individual term.
It is clear that  \eqref{pr:out10-1} and \eqref{pr:out10-3} give \eqref{upR123-1}. 
\end{proof}

\begin{remark}
Note that there is an one-to-one correspondence between the converse proof of Theorem~\ref{thm:mld} and that of Theorem~\ref{thm:md-outer-parametric}. In fact,  here we use the description subsets and their capability of lossy recovering the noisy source layers, where they have been used to losslessly reconstruct the source levels in the A-MLD.
\end{remark}

Now we are ready to prove Theorem~\ref{thm:md-outer}, which is a direct consequence of Theorem~\ref{thm:md-outer-parametric}.

\begin{proof}[Proof of Theorem~\ref{thm:md-outer}] 
We can choose arbitrary values of $d_i$'s, the variance of the additive noise in Theorem~\ref{thm:md-outer-parametric}, such that $d_1\geq d_2\geq \dots\geq d_6>0$. One can optimize the bound in Theorem~\ref{thm:md-outer-parametric} with respect to the values of $d_i$'s, and obtain a bound isolated from $d_i$'s, by replacing them with the optimal choices. Such bound would be the best that can be found using this method. However instead of solving such a difficult optimization problem, we choose $d_i=D_{\L^{-1}(i)}$, for $i=1,\dots,6$. It is clear the $d_i$'s satisfy the desired non-increasing order due to the definition of the ordering level. We will later show that this choice gives a bound which is within constant bit gap from the inner bound  in Theorem~\ref{thm:md-inner}.
 
The single description rate inequalities are exactly the same. The proof of the other inequalities is by straightforward evaluation of their counterparts in Theorem~\ref{thm:md-outer-parametric}, for $d_i=D_{\L^{-1}(i)}$, and applying simple bounds. We do not repeat the same arguments here, and only illustrate such derivation for one simple case. For the sum of two description rates, we can start with \eqref{upR12} and use $d_i=D_{\L^{-1}(i)}$ to get
\begin{align}
R_i+R_j+2\e
&\stackrel{(a)}{\geq}  \frac{1}{2} \log \frac{1+D_{\Gamma_i } }{D_{\Gamma_i}+D_{\Gamma_i}} \frac{1+D_{\Gamma_j}}{D_{\Gamma_j}+D_{\Gamma_j}} \frac{D_{\Gamma_i \Gamma_j} +\min( D_{\Gamma_i} , D_{\Gamma_j}) } {(1+\min( D_{\Gamma_i} , D_{\Gamma_j}) ) D_{\Gamma_i \Gamma_j} }\nonumber\\
&\stackrel{(b)}{=}  
\frac{1}{2}
\log \frac{1+\max( D_{\Gamma_i} , D_{\Gamma_j})  }{4 D_{\Gamma_1} D_{\Gamma_j}}  \frac{D_{\Gamma_i \Gamma_j} +\min( D_{\Gamma_i} , D_{\Gamma_j}) } { D_{\Gamma_i \Gamma_j} }\nonumber\\
&\stackrel{(c)}{\geq}
\frac{1}{2} \log \frac{\min(D_{\Gamma_i} , D_{\Gamma_j} ) }{4D_{\Gamma_i} D_{\Gamma_j} D_{\Gamma_i \Gamma_j}}\nonumber\\
&=-1 +\frac{1}{2} \log\frac{1}{\max (D_{\Gamma_i}, D_{\Gamma_j})} +\frac{1}{2}\log\frac{1}{ D_{\Gamma_i \Gamma_j}},
\label{eqn:R1R2-par}
\end{align}
where we have also used the fact  $d_{\max(a,b)}=\min(d_a,d_b)$ in $(a)$ which is implied by decreasing ordering of $d_i$'s, $(b)$ is due to the fact that $(1+x)(1+y)=(1+\min(x,y))(1+\max(x,y))$, and $(c)$ holds since $D_{\S}$'s are non-negative.  Similar simple manipulations give the other bounds in Theorem~\ref{thm:md-outer}. 

\end{proof}

\subsection{A Simple Coding Scheme for $3$-Description A-MD: Proof of Theorem~\ref{thm:md-inner}}
Our approach to prove Theorem~\ref{thm:md-inner} is to present a simple scheme with description rates satisfying \eqref{oR1}--\eqref{oR123-1} which  guarantees the distortion constraints. This scheme is based on the successive refinability of Gaussian sources \cite{EquitzCover:91, Bixio}, as well as the asymmetric multilevel diversity coding result presented in the previous section. In the encoding scheme, we first produce seven successive refinement layers of the source, and then encode them losslessly. 

\noindent\textbf{Successive refinement coding}\\
Consider the non-increasing sequence of distortion constraints $\D'=(D_{\L^{-1}(1)}, D_{\L^{-1}(2)}, \dots,D_{\L^{-1}(7)})$. Produce seven layers of successive refinement (SR), $\Psi_k$ for $ k=1,2,\dots,7$, such that one can reconstruct the source sequence within distortion constraint $D_{\L^{-1}(k)}$ using $\Psi_1,\dots,\Psi_{k}$. Since the Gaussian source is successively refinable \cite{EquitzCover:91}, it is clear that $\Psi_k$ can be encoded to a binary block  of length arbitrary close to 
\begin{eqnarray}
nh'_k\triangleq nR(D_{\L^{-1}(k)})-nR(D_{\L^{-1}(k-1)})
\label{eq:new_h}
\end{eqnarray} 
where $R(D)=-\frac{1}{2}\log D$ is the unit variance Gaussian R-D function, and $D_{\L^{-1} (0)}\triangleq 1$. Note that by using fixed length code in SR coding, these blocks are block-wise independently and identically distributed.

\noindent\textbf{Multilevel diversity coding}\\
Now, it only remains to produce the descriptions such that the decoder at level $\L(\S)$ can losslessly recover the  precoded bitstream SR layers $\Psi_1,\dots,\Psi_{\L(\S)}$, and then reconstruct the Gaussian source sequence within distortion $D_{\S}$. Encoding and decoding of the precoded SR layers are exactly the A-MLD problem. We can simply use the rate region characterization of the A-MLD problem in Theorem~\ref{thm:mld} to find the achievable rate region of the proposed scheme for A-MD, where only substitution of $V_k=\Psi_k$ and $U_k=(\Psi_1,\dots,\Psi_k)$ is needed.  Therefore we have
\begin{align}
h_k=\frac{1}{2}\log\frac{D_{\L^{-1}(k-1)} }{D_{\L^{-1}(k})}
\end{align}
and 
\begin{align}
H_k=\sum_{j=1}^k h_j=\frac{1}{2}\log\frac{1}{D_{\L^{-1}(k)}}.
\end{align}
Replacing the values of $H_k$'s in Theorem~\ref{thm:mld}, we obtain Theorem~\ref{thm:md-inner}. 

It is worth mentioning that although the successive refinement part of the scheme is well-known, producing the descriptions and their rate characterization is not an easy task without the A-MLD result. As an example, consider a system with ordering level $\L_1$ and assume $D_{\Gamma_2}D_{\Gamma_1\Gamma_3} \leq D_{\Gamma_3}^2\leq D_{\Gamma_2}D_{\Gamma_1 \Gamma_2}$. An achievable rate triple is
\begin{align}
(R_1,R_2,R_3)=(\frac{1}{2}\log\frac{D_{\Gamma_2} D_{\Gamma_2 \Gamma_3}}{D_{\Gamma_1} D_{\Gamma_3} D_{\Gamma_1 \Gamma_2 \Gamma_3}},\frac{1}{2}\log\frac{D_{\Gamma_1 \Gamma_3}}{D_{\Gamma_3}D_{\Gamma_2 \Gamma_3}},\frac{1}{2}\log\frac{D_{\Gamma_3}}{D_{\Gamma_2}D_{\Gamma_1 \Gamma_3}}),
\end{align}
 which corresponds to the corner point $Y_{12}$ in regime II of the A-MLD coding problem. The description encoding for this corner point is illustrated in Fig~\ref{fig:ex-2}. Clearly, the coding scheme for this point matches that for $Y_{12}$ closely, and the SR encoded information in the $3$-rd, $4$-th and $5$-th layers needs to be strategically re-processed using linear codes. Without the underlining A-MLD coding scheme, it appears difficult to devise this coding operation directly.

\begin{figure}[ht]
\begin{center}
	\psfrag{U1}[Bc][Bc]{$\tilde{V}_1$}
	\psfrag{U2}[Bc][bc]{$\tilde{V}_2$}
	\psfrag{U3}[Bc][bc]{$\tilde{V}_{3,1}$}
	\psfrag{U32}[Bc][bc]{$\tilde{V}_{3,2}$}
	\psfrag{U4}[Bc][bc]{$\tilde{V}_4$}
	\psfrag{U5}[Bc][bc]{$\tilde{V}_{5,1}$}
	\psfrag{U52}[Bc][bc]{$\tilde{V}_{5,2}$}
	\psfrag{U6}[Bc][bc]{$\tilde{V}_6$}
	\psfrag{U7}[Bc][bc]{$\tilde{V}_7$}
	\psfrag{l1}[Bc][bc]{$\ell_1$}
	\psfrag{l2}[Bc][bc]{$\ell_2$}
	\psfrag{l31}[Bc][bc]{$\ell_3-\ell_4$}
	\psfrag{l32}[Bc][bc]{$\ell_4$}
	\psfrag{l4}[Bc][bc]{$\ell_4$}
	\psfrag{l51}[Bc][bc]{$\ell_3-\ell_4$}
	\psfrag{l52}[Bc][bc]{$\ell_5+\ell_4-\ell_3$}
	\psfrag{l6}[Bc][bc]{$\ell_6$}
	\psfrag{l7}[Bc][bc]{$\ell_7$}
	\includegraphics[width=150mm]{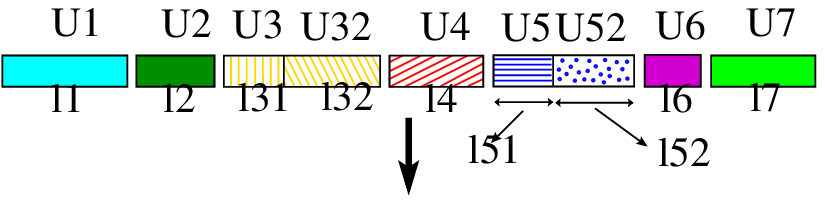}
\end{center}
\begin{center}
\hspace{-25pt}
	\psfrag{g1}[Bc][Bc]{$\Gamma_1:\ $}
	\psfrag{g2}[Bc][bc]{$\Gamma_2:\ $}
	\psfrag{g3}[Bc][bc]{$\Gamma_3:\ $}
	\includegraphics[width=159mm]{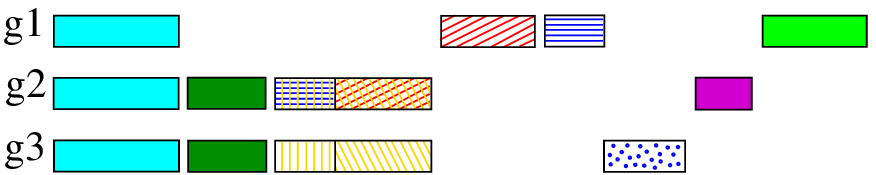}
\end{center}
\caption{Description encoding lossless reconstruction for a system with ordering level $\L_1$ and $D_{\Gamma_2}D_{\Gamma_1\Gamma_3} \leq D_{\Gamma_3}^2\leq D_{\Gamma_2}D_{\Gamma_1 \Gamma_2}$} 
\label{fig:ex-2}
\end{figure}


\section{Conclusion}
We formulated the asymmetric multilevel diversity coding problem, an
asymmetric counterpart for the symmetric version of the problem. A
complete characterization of the admissible rate region is given for
the three-description case. We partition the data and apply
linear network coding (binary \texttt{xor}) on the partitioned subsequences, as
a part of the proposed encoding scheme to achieve the upper bound. It
turns out that using such a strategy of jointly encoding the independent
data streams is crucial, and the outer bound is not achievable without
using it, in contrast to the symmetric problem, in which the
source-separation coding is known to be optimal.

Using the intuition gained through A-MLD coding problem, we consider the Gaussian asymmetric three description problem. Inner and outer bounds for
the admissible rate region are given, and the difference between them
are shown to be bounded by small universal constants. 
Though the general asymmetric Gaussian MD rate distortion
region is hard to characterize, it is satisfying to see that a simple
coding architecture is almost optimal. The A-MLD coding problem plays
a key role in establishing these results, which further strengthens
the connection between the MLD coding and the MD
problem. Philosophically, this work is related to the 
approximation results obtained in the context of the interference and
relay networks \cite{ETW:08, ADT:09}, and further illustrates the effectiveness of 
the general approach of first treating the lossless (deterministic) coding problem,
 and then deriving approximate characterization for its lossy (noisy) counterpart.

\section*{Appendix A}\label{app-proof}

\begin{IEEEproof}[Proof of Lemma~\ref{lemma1}]
\begin{align}
I( \S ; Y_i^n ) &= h(Y_i^n) - h(Y_i^n | \S) \nonumber\\
&= h(Y_i^n) - h(Y_i^n -\hat{X}^n_{\S} | \S) \nonumber\\
&\geq  h(Y_i^n) - h(X^n+Z_i^n -\hat{X}^n_{\S})\nonumber\\
&\geq  h(Y_i^n) - \sum_{t=1}^n h(X(t)-\hat{X}_{\S}(t)+Z_i(t))\nonumber\\
&\stackrel{(a)}{\geq}  \frac{n}{2} \log (2 \pi e (1+d_i)) - \sum_{t=1}^n \frac{1}{2} \log \left( 2\pi e (\E(X(t)-\hat{X}_{\S}(t))^2 + d_i)\right)\nonumber\\
&\stackrel{(b)}{\geq}  \frac{n}{2} \log (1+d_i) - \frac{n}{2} \log \left( \E d(X^n,\hat{X}_{\S}^n) + d_i\right)\nonumber\\
&\stackrel{(c)}{\geq}  \frac{n}{2} \log \frac{1+d_i}{D_{\S} + d_i}
\end{align}
where $(a)$ is due to the fact that the entropy of any random variable is upper bounded by that of a Gaussian variables with the same variance; $(b)$ is implied by   concavity of the function $\log(x)$; and in $(c)$ we have used the fact that $\log(x+a)$ is an increasing function in $x$.  
\end{IEEEproof}

\begin{IEEEproof}[Proof of Lemma~\ref{lemma2}] 
Note that
\begin{align}
h(Y_i^n|\S)-h(Y_j^n|\S)
&\stackrel{(a)}{=} h(Y_i^n|\S)-h(Y_j^n|\S, Z_i^n-Z_j^n)\nonumber\\
&\stackrel{(b)}{=} h(Y_i^n|\S)-h(Y_i^n|\S, Z_i^n-Z_j^n) \nonumber\\
&= I(Y_i^n; Z_i^n-Z_j^n | \S)\nonumber\\
&= h(Z_i^n-Z_j^n | \S) - h(Z_i^n-Z_j^n | \S, Y_i^n)\nonumber\\
&\stackrel{(c)}{\geq} h(Z_i^n-Z_j^n) - h(Z_i^n-Z_j^n | Y_i^n-\hat{X}_{\S}^n)\nonumber\\
&\geq \sum_{t=1}^n h(Z_i(t)-Z_j(t)) - h(Z_i(t)-Z_j(t) | Y_i(t)-\hat{X}_{\S}(t))\nonumber\\
&= \sum_{t=1}^n I(Z_i(t)-Z_j(t) ; X(t) -\hat{X}_{\S}(t) + Z_i(t))\nonumber\\
&\stackrel{(d)}{\geq} \sum_{t=1}^n \frac{1}{2}\log\frac{\E(X(t)-\hat{X}_{\S}(t))^2+d_i}{\E(X(t)-\hat{X}_{\S}(t))^2+d_j}\nonumber\\
&\stackrel{(e)}{\geq} \frac{n}{2} \log \frac{D_{\S}+d_i}{D_{\S}+d_j}
\end{align}
where $(a)$ holds because $Y_j^n$ is independent of $Z_i^n-Z_j^n=N_i^n+\cdots+N_{j-1}^n$ for $i<j$; the equality in $(b)$ is because of  $Y_i=Y_j+(Z_i-Z_j)$; $(c)$ is due to the data processing inequality and the fact that $Z_i^n-Z_j^n$ is purely noise and independent of $X^n$ and therefore $\S$; in $(d)$ we use the worst noise lemma in \cite{CoverThomas,DigCov01}; and $(e)$ is due to convexity and monotonicity of $\log (x+a)/(x+b)$ in $x$ when $a\geq b$. 
Therefore, we simply have
\begin{align}
I(\S ; Y_j^n) -I(\S ; Y_i^n) &= h(Y_j^n) -h(Y_j^n|\S) -h(Y_i^n) + h(Y_i^n|\S)\nonumber\\
&\geq  \frac{n}{2} \log\frac{1+d_j}{1+d_i} + \frac{n}{2} \log \frac{D_{\S}+d_i}{D_{\S}+d_j}\nonumber\\
&=\frac{n}{2} \log \frac{(1+d_j)(D_{\S}+d_i)}{(1+d_i)(D_{\S}+d_j)}.
\end{align}
\end{IEEEproof}

\end{document}